\documentclass[letterpaper,11pt]{article}
\usepackage{microtype}
\usepackage[letterpaper,margin=1in]{geometry}
\usepackage[numbers,sort&compress]{natbib}
\usepackage{xcolor}
\usepackage{amsmath}
\usepackage{amssymb}
\usepackage{mathrsfs}
\usepackage{amsthm}
\usepackage{enumitem}
\usepackage{textgreek}
\usepackage{graphicx}
\usepackage{framed}
\usepackage[small]{caption}
\usepackage{booktabs}
\usepackage{tikz-cd}
\usepackage[unicode]{hyperref}
\usepackage{doi}
\usepackage[textsize=tiny]{todonotes}
\usepackage[framemethod=tikz]{mdframed}
\usepackage{thm-restate}
\usepackage{mathtools}
\usepackage{xspace}

\theoremstyle{plain}
\newtheorem{theorem}{Theorem}
\newtheorem{lemma}[theorem]{Lemma}

\newtheorem{corollary}[theorem]{Corollary}

\newtheorem{observation}[theorem]{Observation}

\theoremstyle{definition}
\newtheorem{definition}[theorem]{Definition}

\theoremstyle{remark}

\newcommand{\namedref}[2]{\hyperref[#2]{#1~\ref*{#2}}}

\newcommand{\s}{\mspace{1mu}}
\newcommand{\mybox}[1]{\mspace{2mu}{\setlength{\fboxsep}{1.5pt}\color{lightgray}\boxed{\color{black}\scriptstyle #1}}\mspace{2mu}}
\newcommand{\A}{\mathsf{A}}
\newcommand{\B}{\mathsf{B}}
\renewcommand{\C}{\mathsf{C}}
\renewcommand{\L}{\mathsf{L}}

\renewcommand{\a}{\mathsf{a}}
\renewcommand{\b}{\mathsf{b}}

\newcommand{\y}{\mathsf{y}}

\newcommand{\Z}{\mathsf{Z}}

\newcommand{\M}{\mathsf{M}}

\renewcommand{\P}{\mathsf{P}}
\renewcommand{\O}{\mathsf{O}}

\newcommand{\X}{\mathsf{X}}
\newcommand{\Y}{\mathsf{Y}}

\newcommand{\id}{\mathsf{Id}}

\newcommand{\xx}{\mathsf{X}}
\newcommand{\xo}{\mathsf{O}}
\newcommand{\xa}{\mathsf{A}}
\newcommand{\xp}{\mathsf{P}}
\newcommand{\zx}{\mathsf{M}}
\newcommand{\zo}{\mathsf{U}}
\newcommand{\za}{\mathsf{B}}
\newcommand{\zp}{\mathsf{Q}}

\DeclareMathOperator{\re}{\mathcal R}
\DeclareMathOperator{\rere}{\overline{\mathcal R}}

\newcommand{\nodeconst}{\ensuremath{\mathcal{N}}}
\newcommand{\edgeconst}{\ensuremath{\mathcal{E}}}

\newcommand{\set}[1]{\left\{ #1 \right\}}

\newcommand{\LOCAL}{\ensuremath{\mathsf{LOCAL}}\xspace}
\newcommand{\CONGEST}{$\mathsf{CONGEST}$\xspace}

\DeclareMathOperator{\poly}{poly}

\newenvironment{myabstract}
{\list{}{\listparindent 1.5em%
		\itemindent    \listparindent
		\leftmargin    1cm
		\rightmargin   1cm
		\parsep        0pt}%
	\item\relax}
{\endlist}

\newenvironment{mycover}
{\list{}{\listparindent 0pt
		\itemindent    \listparindent
		\leftmargin    1cm
		\rightmargin   1cm
		\parsep        0pt}%
	\raggedright
	\item\relax}
{\endlist}

\newcommand{\myemail}[1]{\,$\cdot$\, {\small #1}}
\newcommand{\myaff}[1]{\,$\cdot$\, {\small #1}\par\smallskip}

\hypersetup{
	colorlinks=true,
	linkcolor=black,
	citecolor=black,
	filecolor=black,
	urlcolor=[rgb]{0,0.1,0.5},
	pdftitle={Improved Distributed Lower Bounds for MIS and Bounded (Out-)Degree Dominating Sets in Trees},
	pdfauthor={Alkida Balliu, Sebastian Brandt, Fabian Kuhn, Dennis Olivetti}
}

\begin{document}

	\begin{mycover}
		{\huge\bfseries\boldmath Improved Distributed Lower Bounds for MIS and Bounded (Out-)Degree Dominating Sets in Trees \par}
		\bigskip
		\bigskip
		\bigskip
		
		\textbf{Alkida Balliu}
		\myemail{alkida.balliu@cs.uni-freiburg.de}
		\myaff{University of Freiburg}
		
		\textbf{Sebastian Brandt}
		\myemail{brandts@ethz.ch}
		\myaff{ETH Zurich}
		
		\textbf{Fabian Kuhn}
		\myemail{kuhn@cs.uni-freiburg.de}
		\myaff{University of Freiburg}
		
		\textbf{Dennis Olivetti}
		\myemail{dennis.olivetti@cs.uni-freiburg.de}
		\myaff{University of Freiburg}
	\end{mycover}
	\bigskip

\begin{myabstract}
	
	Recently, Balliu, Brandt, and Olivetti [FOCS '20] showed the first $\omega(\log^* n)$ lower bound for the maximal independent set (MIS) problem in trees. In this work we prove lower bounds for a much more relaxed family of distributed symmetry breaking problems. As a by-product, we obtain improved lower bounds for the distributed MIS problem in trees.
	
	For a parameter $k$ and an orientation of the edges of a graph $G$, we say that a subset $S$ of the nodes of $G$ is a \emph{$k$-outdegree dominating set} if $S$ is a dominating set of $G$ and if in the induced subgraph $G[S]$, every node in $S$ has outdegree at most $k$. Note that for $k=0$, this definition coincides with the definition of an MIS. For a given $k$, we consider the problem of computing a $k$-outdegree dominating set. We show that, even in regular trees of degree at most $\Delta$, in the standard \LOCAL model, there exists a constant $\epsilon>0$ such that for $k\leq \Delta^\epsilon$, for the problem of computing a $k$-outdegree dominating set, any randomized algorithm requires at least $\Omega\Big(\min\set{\log\Delta,\sqrt{\log\log n}}\Big)$ rounds and any deterministic algorithm requires at least $\Omega\Big(\min\set{\log\Delta,\sqrt{\log n}}\Big)$ rounds.

  	The proof of our lower bounds is based on the recently highly successful round elimination technique. We provide a novel way to do simplifications for round elimination, which we expect to be of independent interest. Our new proof is considerably simpler than the lower bound proof in [FOCS '20]. In particular, our round elimination proof uses a family of problems that can be described by only a constant number of labels. The existence of such a proof for the MIS problem was believed impossible by the authors of [FOCS '20].
\end{myabstract}

\thispagestyle{empty}
\setcounter{page}{0}
\newpage
\newpage

\section{Introduction}
\label{sec:intro}

The general question of understanding the distributed complexity of basic symmetry breaking tasks is at the heart of the area of distributed graph algorithms. Formally, distributed symmetry breaking problems are modeled as locally checkable problems, and they are typically studied in the standard \LOCAL model~\cite{Linial1992,Peleg2000}. In the \LOCAL model, the nodes of a graph $G=(V,E)$ communicate in synchronous rounds by exchanging possibly arbitrarily large messages with their neighbors. When terminating, an algorithm has to assign a label to each node and/or edge of the graph such that in every local neighborhood, a condition given by the considered locally checkable problem is satisfied. A prototypical example of such a distributed symmetry breaking task is the problem of computing a maximal independent set (MIS) of a graph $G$. Here, each node has to either output $1$ (in the MIS) or $0$ (not in the MIS) such that no two neighbors output $1$ and such that every node that outputs $0$ has at least one neighbor that outputs $1$. The distributed complexity of computing an MIS has been studied intensively for more than 30 years and understanding the complexity of computing an MIS and of closely related problems is at the core of understanding the complexity of symmetry breaking problems more generally (see., e.g.,~\cite{Barenboim2013,stoc17complexity,Linial1992,Luby1985}). 

The objectives of this paper are twofold. On the one hand, we want to improve our understanding of the distributed complexity of computing an MIS and of related graph structures. On the other hand, we also want to more generally improve our understanding of lower bound techniques for the \LOCAL model and in particular, we make further developments on the recently highly successful round elimination technique. Specifically, in this paper, we study the problem of computing an MIS and of computing bounded degree and bounded outdegree dominating sets in tree networks. Bounded degree and outdegree dominating sets are natural relaxations of an MIS, which we define next.

\paragraph{Bounded (Out-)Degree Dominating Sets.} An MIS of a graph $G=(V,E)$ is a node set $S$ such that $S$ is a \emph{dominating set} (every node node in $S$ has a neighbor in $S$) and an \emph{independent set} (no two neighbors are in $S$) of $G$. Two particularly natural ways to relax the requirement of an MIS are therefore to either relax the domination condition or to relax the independence condition. We can relax the domination condition by requiring that each node $v\not\in S$ only has a node in $S$ within some fixed distance $r\geq 1$. We then obtain what is known as a $(2,r)$-ruling set~\cite{Awerbuch89}. The distributed complexity of computing such ruling sets has been studied quite extensively~\cite{Awerbuch89,balliurules,Barenboim2012,Gfeller07,ghaffari16improved,KuhnMW18}. In this paper, we relax the MIS requirement differently. We keep the domination condition and we instead relax the independence condition.  For a graph $G=(V,E)$ and an integer parameter $k\geq 0$, a node set $S\subseteq V$ is called a \emph{$k$-degree dominating set} if $S$ is a dominating set of $G$ and if the induced subgraph $G[S]$ has maximum degree at most $k$. Further, if we are also given an orientation of the edges of the induced subgraph $G[S]$, then $S$ together with this edge orientation is called a \emph{$k$-outdegree dominating set} if $S$ is a dominating set and every node $v\in S$ has outdegree at most $k$ in $G[S]$. Note that for $k=0$, both definitions coincide with the definition of an MIS and for $k>0$, $k$-degree and $k$-outdegree dominating sets are thus a natural relaxation of an MIS. While we are not aware of any work on distributed algorithms for computing bounded degree or outdegree dominating sets, there is previous work on closely related problems on line graphs. An MIS of the line graph of a graph $G$ is a maximal matching of $G$. A natural generalization of matchings on graphs are so-called $b$-matchings, where in its simplest form, a $b$-matching is a set of edges of a graph such that no node is contained in more than $b$ edges. The distributed computation of (variations of) $b$-matchings has for example been studied in \cite{Balliu2019,trulytight,fischer17improved,Koufogiannakis09}.



\subsection{Our Results}
\label{sec:results}

As our main technical contribution, for any $k\leq \Delta^\epsilon$ for some constant $\epsilon>0$, we prove an $\Omega(\log\Delta)$-round lower bound for computing $k$-outdegree dominating sets in $\Delta$-regular trees in the port numbering model.\footnote{In the port numbering model, the nodes of the communication graph do not have unique IDs, but they can distinguish their neighbors by distinct port numbers.} The lower bound for the port numbering model is then lifted to the more powerful general \LOCAL model by using standard techniques developed in \cite{balliurules,binaryLCL,trulytight,Brandt2016,Brandt2019,chang16exponential}, leading to the following main result.
\begin{theorem}\label{thm:main}
	For $k \le \Delta^\epsilon$ and some constant $\epsilon > 0$, the $k$-outdegree dominating set problem requires time $\Omega(\min\{\log \Delta, \log_\Delta n\})$ in the deterministic \LOCAL model and time $\Omega(\min\{\log \Delta, \log_\Delta \log n\})$ in the randomized \LOCAL model, in $\Delta$-regular trees of $n$ nodes.
\end{theorem}

By choosing the right value of $\Delta$, we obtain the following corollary.
\begin{corollary}\label{cor:main}
  	For some constant $\epsilon > 0$, the $k$-outdegree dominating set problem requires  at least $\Omega\Bigl(\min\set{\log \Delta, \sqrt{\log n}}\Bigr)$ rounds in the deterministic \LOCAL model for $k \le \min\{\Delta^\epsilon,2^{\epsilon \sqrt{\log n}}\}$, and $\Omega\Bigl(\min\set{\log \Delta, \sqrt{\log\log n}}\Bigr)$ rounds in the randomized \LOCAL model for $k \le \min\{\Delta^\epsilon,2^{\epsilon \sqrt{\log \log n}}\}$, in $n$-node trees of maximum degree at most $\Delta$.
\end{corollary}

Note that the same lower bound of course also holds for the $k$-degree dominating set problem as a $k$-degree dominating set can be transformed into a $k$-outdegree dominating set by orienting the edges in an arbitrary way. For the problem of computing an MIS, the lower bound of Theorem \ref{thm:main} improves and in our view also significantly simplifies a lower bound from \cite{balliurules}, where it was shown that computing an MIS in $n$-node trees of maximum degree at most $\Delta$ requires at least $\Omega\Bigl(\min\set{\frac{\log \Delta}{\log\log\Delta}, \sqrt{\frac{\log n}{\log\log n}}}\Bigr)$ in the deterministic \LOCAL model and time $\Omega\Bigl(\min\set{\frac{\log\log \Delta}{\log\log\log\Delta}, \sqrt{\frac{\log\log n}{\log\log\log n}}}\Bigr)$ in the randomized \LOCAL model. While the lower bound proof of \cite{balliurules} also provides some lower bounds for computing ruling sets, our new lower bound proof generalizes to computing bounded degree and outdegree dominating sets. We note that for general graphs, there are stronger lower bounds known. In particular, in \cite{Balliu2019,trulytight}, it was shown that for any $b\in[1,\Delta^{1-\epsilon}]$ for a constant $\epsilon>0$, computing a $b$-matching in $\Delta$-regular trees requires time $\Omega\Bigl(\min\set{\frac{\Delta}{b},\frac{\log n}{\log\log n}}\Bigr)$ in the deterministic \LOCAL model and time $\Omega\Bigl(\min\set{\frac{\Delta}{b},\frac{\log\log n}{\log\log\log n}}\Bigr)$ in the randomized \LOCAL model. This immediately implies the same lower bounds for MIS (by setting $b=1$) and for $b$-outdegree and $b$-degree dominating sets in general (regular) graphs. We however believe that it is important to also understand the complexity of MIS and related problems in trees. First, trees are a particularly simple and important family of graphs and at least as a function of the maximum degree $\Delta$, we are not aware of any algorithmic techniques for computing an MIS that work better in trees than in general graphs. Further, MIS lower bounds that hold in trees might also lead to lower bounds on other interesting problems. For example, while computing ruling sets might be as hard in trees as it is in general graphs, the same is provably not true for line graphs, where $(2,r)$-ruling sets for $r\geq 2$ can be computed in $O(\log^* n)$ rounds~\cite{balliurules,KuhnMW18}.

\paragraph{Distributed Algorithms for Bounded (Out-)Degree Dominating Sets.} Before we discuss the technical details of our approach in more detail, we want to briefly discuss what we can say about upper bounds on the complexity of computing $k$-(out-)degree dominating sets. We here concentrate on running times as a function of the maximum degree $\Delta$ of the $n$-node input graph $G$. For $k=0$, both problems ask for computing an MIS, which can be done deterministically in time $O(\Delta+\log^* n)$ by using an algorithm from \cite{barenboim14distributed}. The lower bounds of \cite{Balliu2019} imply that at least for general graphs, this is best possible even for randomized algorithms, unless we allow a significantly larger dependency of the time complexity on the number of nodes $n$. For $k\geq 1$, the fastest algorithms that we are aware of are based on the following simple idea. As a first step, we compute a defective or an arbdefective coloring of the graph. A $k$-defective $c$-coloring of $G$ partitions the nodes of $G$ into $c$ color classes such that the maximum degree of the induced subgraph of each color class is at most $k$. A $k$-arbdefective $c$-coloring is a partition of the nodes into $c$ color classes together with an orientation of the edges such that the maximum outdegree of the induced subgraph of each color class is at most $k$. From a $k$-defective or $k$-arbdefective $c$ coloring, we can get to $k$-degree or $k$-outdegree dominating set as follows. We start with an empty set $S$ and iterate over the $c$ color classes. When considering the nodes of a given color class, we add all nodes to the set $S$ that do not already have a neighbor in $S$. It is known that a $k$-arbdefective coloring with $O(\Delta/k)$ colors can be computed in time $O(\Delta/k + \log^* n)$~\cite{barenboim18} and that a $k$-defective coloring with $O((\Delta/k)^2)$ colors can be computed in $O(\log^* n)$~\cite{Kuhn09}. This implies that we can always compute a $k$-outdegree dominating set in time $O(\Delta/k + \log^*n)$ and that we can compute a $k$-degree dominating set in time $O\big(\min\set{\Delta,(\Delta/k)^2}+\log^* n\big)$. Note that the running time for $k$-outdegree dominating sets matches the lower bound of \cite{trulytight} for general graphs. In fact, the lower bound in \cite{trulytight} holds for line graphs. In line graphs, any $k$-outdegree dominating set $S$ is directly also an $O(k)$-degree dominating set. This follows because if some edge $\set{u,v}$ in $S$ has $k$ outneighbors in an oriented version of the line graph, then either $u$ or $v$ must be contained in $\Omega(k)$ edges of $S$. For line graphs, we therefore know that the lower bound of \cite{trulytight} is asymptotically tight both for $k$-outdegree and for $k$-degree dominating sets (unless we allow a larger dependency of the time on $n$).

\subsection{Our Approach and Techniques}
In order to prove our results we use the automatic version of the round elimination technique~\cite{Brandt2019,Olivetti2019}. On a high level, the round elimination technique works as follows. Suppose we want to prove a lower bound for a problem $\Pi_0$ of interest. The idea is to find a \emph{lower bound sequence} of problems $\Pi_0\rightarrow\Pi_1\rightarrow\Pi_2\rightarrow\ldots$ such that each problem $\Pi_i$ is solvable at least one round faster than problem $\Pi_{i-1}$ (as long as $\Pi_{i-1}$ is not $0$-rounds solvable). If we can find such a sequence and in addition we can show that problem $\Pi_{T-1}$ is not $0$-round solvable, then we have proven a lower bound of $T$ rounds for our problem of interest $\Pi_0$. The automatic round elimination technique is able to theoretically provide a lower bound sequence of problems in an automatic way. Unfortunately it often happens that, if we start from a problem of interest, such as MIS, and apply this technique as is, the size of the description of the obtained problems grows roughly doubly exponentially in each step of the problem sequence. Previous works that have used the automatic round elimination technique usually tried to keep the description of the problems small enough by carefully \emph{simplifying} each problem in the sequence: a simplification must be such that, on the one hand it reduces the size of the problem description, and on the other hand it must not make the problem too easy to solve, which would result in a trivial lower bound. The round elimination technique can be used for finding upper bounds as well. For this purpose, the idea is to find an \emph{upper bound sequence} of problems $\Pi_0\rightarrow\Pi_1\rightarrow\Pi_2\rightarrow\ldots$ such that each problem $\Pi_i$ is solvable at most one round faster than problem $\Pi_{i-1}$. Then, if we can show that problem $\Pi_T$ is $0$-round solvable, we get an upper bound of $T$ rounds for our problem of interest $\Pi_0$. Generally, there have been mainly $4$ different approaches for applying round elimination for proving lower bounds, based on the number of labels in the problem sequence.
\begin{itemize}
	\item \emph{Doubly exponential growth}: A possible approach is to embrace the doubly exponential growth of each problem in the lower bound sequence and let the labels grow without trying to understand the structure of the problems. Usually, in this way one can show only $\Omega(\log^* n)$ or $\Omega(\log^* \Delta)$ lower bounds (see e.g., \cite{Linial1992, Brandt2019}).
	
	\item \emph{Similarity}: Another approach consists in making careful simplifications in order to try to make each problem as similar as possible to the previous problem of the sequence. This usually allows us to obtain a family where each problem has the same number of labels, that is, the number of labels stays constant, making it easier to apply the round elimination technique. Unfortunately, sometimes this approach seems not to work, and by making this kind of simplifications we obtain problems that are too easy. This approach has been used for example in~\cite{Balliu2019}.
	
	\item \emph{Fixed points}: A case limit of the similarity approach is the one in which we try to build a lower bound sequence where there is a non-$0$-round solvable problem that appears more than once in the sequence (i.e., a fixed point). This directly implies an $\Omega(\log n)$ deterministic and $\Omega(\log\log n)$ randomized lower bound (see e.g., \cite{binaryLCL}).
	
	\item \emph{Linear growth}: Sometimes, the similarity approach does not seem to work, and if the number of labels grows exponentially we do not get a sufficiently strong lower bound. We can then try to let the number of labels grow slower than exponentially, e.g., linearly. Unfortunately, in this case it seems hard to guess the right problem sequence such that we can then indeed prove that it is a lower bound sequence. In~\cite{balliurules}, the following interesting approach was used. First, the authors find an upper bound sequence for the problem of interest, and they then turn this sequence into a lower bound one by performing some suitable simplifications.
\end{itemize}
The MIS problem, and more generally the $k$-outdegree dominating set problems, fall into the category of those problems that grow roughly doubly exponentially at each step of round elimination. Hence, in order to find a polylogarithmic lower bound for such problems we must find a way to reduce the size of the description of each problem in the lower bound sequence. Ideally, we would like to find a lower bound sequence such that we can describe each problem with a constant the number of labels, since this would make the approach considerably simpler and cleaner. Regarding if such a lower bound sequence exists or not for the MIS problem, Balliu, Brandt, and Olivetti~\cite{balliurules} stated the following.

\begin{framed}
	\noindent While we do not have a proof, we do not believe that for MIS such a sequence yielding a polylogarithmic lower bound exists.
\end{framed}

We disprove this statement and we show polylogarithmic lower bounds for MIS and more generally for $k$-outdegree dominating sets, by finding a lower bound sequence where each problem is described using only a constant number of labels. The key idea of our novel approach relies on assuming to have a $\Delta$-edge coloring as input, and then we exploit this input for performing some simplifications that could not be done without assuming such an input (note that this only strengthens our lower bounds, which hold even in the case where we have such an input).

\paragraph{Approach.} Let us discuss the high level idea of our approach in more detail. For simplicity, we focus on the case of the MIS problem (the lower bound for the more general $k$-outdegree dominating sets uses a very similar approach). We start from a problem $\Pi_0$ that is a relaxation of the problem of computing an MIS and we generate a problem sequence $\Pi_0,\Pi_1,\Pi_2,\dots$ such that informally, each of the problems $\Pi_i$ is a mix between an MIS and a problem that requires nodes to output a good orientation. In particular, in each problem $\Pi_i$ we have $3$ types of nodes, which, slightly simplified, take over the following roles. 
\begin{itemize}
	\item \emph{IS-nodes}: nodes in the independent set ($S$).
	\item \emph{Orientation-nodes}: nodes not in $S$ that are required to have a certain number of outgoing-oriented incident edges.
	\item \emph{Non-IS-nodes}: nodes not in $S$ that must have a neighbor that is either an IS-node or an orientation-node.
\end{itemize}
Let $\Pi_1$ be the problem that we get by performing one step of the round elimination technique on $\Pi_0$. What we would like to do is to prove that $\Pi_1$ is such that we can use its solution to solve in $0$ rounds some problem that looks very similar to $\Pi_0$, that is, we would like to use the \emph{similarity} approach described before. Unfortunately, this seems to be not possible, since $\Pi_1$ seems to contain some additional allowed configurations that cannot be relaxed to the original ones without making the obtained problem too easy. Here lies the novelty of our approach: we show that, by exploiting some input given to the nodes, then it becomes possible to achieve our goal. In particular, we prove that $\Pi_1$ is such that we can use its solution, plus a $\Delta$-edge coloring given in input to the nodes, to solve in $0$ rounds a variant $\Pi'_0$ of $\Pi_0$ where the number of required outgoing edges of orientation-nodes in $\Pi'_0$ is less than that of orientation-nodes in $\Pi_0$, but not by much. We then iterate the above reasoning and show that, at each step, the number of required outgoing edges of orientation-nodes goes down by at most a constant factor. This means that we can repeat this process $\Omega(\log \Delta)$ times before getting a problem that is too easy to solve. In other words, with the above process we build a lower bound sequence of length $\Omega(\log \Delta)$ such that the last problem is not $0$-round solvable, hence proving our lower bound.

\subsection{Additional Related Work}

\paragraph{Round Elimination Technique.}
 Linial's $\Omega(\log^* n)$ lower bound for $3$-coloring a cycle was the first lower bound that was proven by what can be understood as a version of round elimination~\cite{Linial1992}. In its current form, round elimination was introduced in a seminal paper by Brandt et al.~\cite{Brandt2016}, which proves a randomized $\Omega(\log\log n)$-round lower bound on the problems of computing a sinkless orientation or a $\Delta$-coloring and as a corollary more generally for the distributed constructive Lova\'{a}sz Local Lemma problem. The lower bound of \cite{Brandt2016} was later lifted to an $\Omega(\log n)$ deterministic lower bound in \cite{chang16exponential} and it was generalized to the $(2\Delta-2)$-edge coloring problem in \cite{chang18complexity}. Later, the round elimination technique was used to show an $\Omega(\log^*n)$ lower bound for the weak $2$-coloring problem~\cite{BalliuHOS19}. In 2019, Brandt~\cite{Brandt2019} refined the technique presenting an \emph{automatic} way to perform round elimination. Based on this result, Olivetti~\cite{Olivetti2019} implemented a tool, called round eliminator, that performs round elimination automatically. These two works were fundamental for a better understanding of the round elimination technique. In fact, Balliu et al.~\cite{Balliu2019} managed to apply the automatic round elimination technique to prove intriguing novel lower bounds for maximal matchings and for MIS. These results were afterwards improved and generalized by Brandt and Olivetti~\cite{trulytight}. Then, round elimination was used for giving a complete characterization of locally checkable problems that can be encoded in the edge-formalism of the round elimination framework by using at most two labels~\cite{binaryLCL}. Fraigniaud and Paz showed how to use round elimination with the algebraic topology framework~\cite{FraigniaudPaz20}. Recently, Balliu, Brandt, and Olivetti~\cite{balliurules} used round elimination to show lower bounds for ruling sets on trees. They also showed that round elimination can be used to obtain upper bounds for ruling sets that asymptotically match the existing best known upper bounds.

\paragraph{Distributed MIS Algorithms.} 
The MIS problem is one of the most extensively studied symmetry breaking problems in the distributed setting (e.g., \cite{Awerbuch89, KarpW85, Luby1985, Alon1986, Linial1992, Naor1991, panconesi96decomposition, Kuhn2004, BarenboimE10, SchneiderW10, LenzenW11, Barenboim2016, Barenboim2013, barenboim14distributed, ghaffari16improved, Rozhon2020,GGR2020}). The first works date back to the late 1980s, where Luby \cite{Luby1985} and Alon, Babai, and Itai \cite{Alon1986} devised parallel $O(\log n)$-round randomized algorithms for solving MIS on graphs with $n$ nodes. Those algorithm directly also work in the distributed setting. The first deterministic distributed MIS algorithm was also developed in the late 1980s. In \cite{Awerbuch89}, Awerbuch et al.\ introduced a generic tool known as network decomposition that allows to deterministically compute an MIS in time $2^{O(\sqrt{\log n\log\log n})}$ in the \LOCAL model. This was later improved to a time of $2^{O(\sqrt{\log n})}$  by Panconesi and Srinivasan~\cite{panconesi96decomposition}. The first improvement on the randomized MIS complexity was obtained by Barenboim et al.~\cite{Barenboim2016}, who showed that an MIS can always be computed in time $O(\log^2\Delta) + 2^{O(\sqrt{\log\log n})}$. The algorithm is faster than the simple $O(\log n)$-time algorithms if the maximum degree $\Delta$ is moderately small. The paper introduced what is now known as the shattering technique to the area of distributed algorithms. First, the problem is solved on most of the graph by using some randomized method, and afterwards the remaining small unsolved components are solved by using the fastest known deterministic algorithm. The result of \cite{Barenboim2016} was improved by Ghaffari~\cite{ghaffari16improved}, who showed that an MIS can be computed in time $O(\log\Delta)+2^{O(\sqrt{\log\log n})}$ in the randomized \LOCAL model. Further improvements were obtained through a recent breakthrough by Rozhon and Ghaffari~\cite{Rozhon2020}, who showed that a network decomposition as introduced in \cite{Awerbuch89} can actually be computed in $\poly\log n$ time deterministically in the \LOCAL model. A slight improvement of the algorithm of \cite{Rozhon2020} leads to the current best deterministic complexity of $O(\log^5 n)$ and the currently best randomized complexity of $O(\log\Delta + \log^5\log n)$ for computing an MIS in the \LOCAL model~\cite{GGR2020}. In a separate line of work, it was shown by Barenboim, Elkin, and Kuhn in \cite{barenboim14distributed} that MIS can be computed in time $O(\Delta+\log^* n)$. Note that this is faster than the previously mentioned algorithms if the maximum degree $\Delta$ is sufficiently small. In addition, the problem of computing an MIS has also been studied in the more restrictive \CONGEST model, where in each round, every node can only send an $O(\log n)$-bit message to each neighbor~\cite{DISC17_MIS,Ghaffari19congest,ghaffariPortman19,GGR2020,Rozhon2020}. Together, these papers imply that also in the \CONGEST model, an MIS can be computed deterministically in time $O(\log^5 n)$ and the fastest known randomized algorithm has a time complexity of $O(\log\Delta\cdot\log\log n + \log^6\log n)$ and is therefore almost as fast as the fastest known randomized \LOCAL model algorithm. 

To conclude our discussion of distributed MIS algorithms, we finally also want to discuss previous work on solving the MIS problem on trees. The first such paper is by Barenboim and Elkin~\cite{BarenboimE10}. They show that in graphs of bounded arboricity and therefore in particular in trees, an MIS can be computed deterministically in only $O(\log n / \log\log n)$ rounds. This still is the fastest known deterministic algorithm for computing an MIS in an $n$-node tree. The first paper to explicitly look at randomized algorithms for computing MIS in trees is by Lenzen and Wattenhofer~\cite{LenzenW11}. They show that in trees, an MIS can be computed in $O(\sqrt{\log n}\log\log n)$ randomized rounds. This result was later improved by Barenboim et al.~\cite{Barenboim2016} and by Ghaffari~\cite{ghaffari16improved} to $O(\sqrt{\log n\log\log n})$ rounds and to $O(\sqrt{\log n})$ rounds, respectively. Those time complexities can be improved if the maximum degree $\Delta$ is moderately small. In this case, the fastest known randomized MIS algorithm on trees has a round complexity of $O(\log\Delta + \log\log n/\log\log\log n)$~\cite{ghaffari16improved}.

\paragraph{Distributed MIS Lower Bounds.} 
 On the lower bound side, it has been known since the late 1980s and early 1990s that computing an MIS requires $\Omega(\log^* n)$ rounds~\cite{Linial1992, Naor1991}, even on graphs with maximum degree $2$ and even for randomized algorithms. Much later, Kuhn, Moscibroda, and Wattenhofer~\cite{Kuhn2004} proved that computing an MIS requires $\Omega\Bigl(\min\Bigl\{\frac{\log\Delta}{\log\log\Delta},\sqrt{\frac{\log n}{\log\log n}}\Bigr\}\Bigr)$ rounds even for randomized algorithms. Balliu et al.~\cite{Balliu2019} recently improved and complemented this result by showing that computing an MIS requires $\Omega\Bigl(\min\Bigl\{\Delta, \frac{\log n}{\log \log n}\Bigr\}\Bigr)$ and $\Omega\Bigl(\min\Bigl\{\Delta, \frac{\log \log n}{\log \log \log n}\Bigr\}\Bigr)$ rounds for deterministic and randomized algorithms, respectively. Except for the $\Omega(\log^* n)$ lower bound, the above MIS lower bounds only hold for general graphs, and they do not apply to sparse graph classes or even to trees. Last year, Balliu, Brandt, and Olivetti~\cite{balliurules} improved the classic $\Omega(\log^*n)$ lower bound also for this case and showed that computing an MIS on regular trees of degree at most $\Delta$ requires $\Omega\left(\min \left\{  \frac{\log \Delta}{\log \log \Delta}  ,  \sqrt{\frac{\log n}{ \log \log n}} \right\} \right)$ and $\Omega\left(\min \left\{  \frac{\log \Delta}{ \log \log \Delta}  , \sqrt{\frac{\log \log n}{ \log \log \log n}} \right\} \right)$ rounds for deterministic and randomized algorithms, respectively.

\section{Preliminaries}

\subsection{The \texorpdfstring{\boldmath \LOCAL}{LOCAL} Model}
The lower bounds that we present in this paper hold for the \LOCAL model of distributed computing. In this model, nodes of a graph $G=(V,E)$, where $|V|=n$, have a globally unique identifier from $\{1,2,\dotsc, \poly(n)\}$. Initially, each node $v$ knows its own identifier $\id(v)$, its own degree $\deg(v)$, the maximum degree $\Delta$ of the graph, and the total number $n$ of nodes. The computation proceeds in synchronous rounds. In each round, each node sends messages to its neighbors, receives messages from its neighbors, and performs local computation. Both the size of the messages and the local computation are not bounded, i.e., messages can be of arbitrary size and the local computation can be of arbitrary complexity. Each node executes the same (distributed) algorithm, and, at some point, each node must decide to terminate its computation. Upon termination, each node must decide its own local output. If the local outputs together form a feasible global output, then we say that the problem is solved correctly. The time complexity of a distributed algorithm is measured as the number of rounds required until all nodes terminate. Due to the unbounded size of the messages, a distributed algorithm of time complexity $T$ for the \LOCAL model can be seen as a function that maps $T$-radius neighborhoods into outputs. 

In the randomized version of the \LOCAL model, in addition to the above, nodes have access to a stream of private random bits. We consider Monte Carlo randomized algorithms, that is, the algorithm must provide a global solution that is correct with high probability, that is, with probability at least $1-1/n$.

The \LOCAL model is quite a strong model, hence lower bounds for the \LOCAL model directly apply in other weaker models as well. In fact, another well known model of distributed computation is the \CONGEST model, where the only difference between the two models is regarding the size of the messages: while in the \LOCAL model we do not bound the size of the messages, in the \CONGEST model, the size of the messages is bounded by $O(\log n)$ bits. Hence, the lower bounds that we present in this paper also hold for the \CONGEST model.

\paragraph{The Port Numbering Model.}
While our model of interest is the \LOCAL model of distributed computation, for technical reasons, we first show our results in a weaker model, that is the \emph{port numbering} (PN) model, and then we then lift them to the \LOCAL model. The only difference between the PN model and the \LOCAL model is that, in the former, nodes do not have an identifier, instead they have a \emph{port numbering} in input. More precisely, in the PN model, each node $v$ has a port number in $\{1,\dotsc,\deg(v)\}$ assigned to each of its incident edges, and this port numbering is such that the incident edges of a node have pairwise different ports assigned to them. As in the case of the randomized \LOCAL model, in the randomized PN model each node has access to a stream of private random bits, and a randomized algorithm  must provide a global solution that is correct with high probability. For technical reasons, we will assume that edges are equipped with a port numbering as well, that is, each edge is provided with a port number in $\{1,2\}$ for its two endpoints, determining a consistent orientation of the edges. This is just a technical detail that makes the PN model only stronger, which may only make the task of proving lower bounds harder.

\subsection{Problems}\label{subsec:problems}

As already mentioned, we will use the round elimination framework to show our lower bounds. This formalism is expressive enough to include \emph{locally checkable} problems, as long as the description of these problems does not involve small cycles. Locally checkable problems are such that their solution is globally correct if it is locally correct in all constant-radius neighborhoods. (For simplicity, in the remainder of the paper we use the term ``locally checkable problems'' to refer to all those problems that are locally checkable and that do not involve small cycles.) Before going into details of the round elimination technique, we first see how problems are defined in this framework, what it means to correctly solve such problems, and how we represent them. Then, as an example, we see how to encode the problem of computing an MIS in this framework.

\paragraph{Characterization of a Problem in the Framework.}  For simplicity, we will only focus on problems defined on $\Delta$-regular graphs, as this will be enough for the purposes of our paper. Such problems, in the round elimination framework, are characterized by the following three objects.
\begin{itemize}
	\item An alphabet $\Sigma$ of possible labels; a \emph{configuration} is a word over the alphabet $\Sigma$.
	\item \emph{Node constraint} $\nodeconst$: a collection of configurations of length $\Delta$.
	\item \emph{Edge constraint} $\edgeconst$: a collection of configurations of length $2$.
\end{itemize}
The order of the elements of a configuration does not matter, and the same label may appear several times in a configuration. We call a configuration in $\nodeconst$ a \emph{node configuration}, and similarly we call a configuration in $\edgeconst$ an \emph{edge configuration}.

\paragraph{Solving a Problem in the Framework.} Given a graph $G=(V,E)$ and the set $A=\{(v,e)\in V\times E~|~ v\in e\}$ of all (node, incident edge) pairs, we say that an algorithm correctly solves a problem characterized by $\Sigma$, $\nodeconst$, $\edgeconst$, if, (1) for each node $v\in V$ it holds that $v$ has assigned an element of $\Sigma$ on each incident edge $(v,e)\in A$ forming a configuration in $\nodeconst$, and (2) for each edge $e=(u,v)$ it holds that the elements on the endpoints $(u,e), (v,e)$ form a configuration in $\edgeconst$.

\paragraph{Representation of Problems in the Framework.} We represent the node and edge configurations using regular expressions. For example, we will denote with $\P\s\O^{\Delta-1}$ a node configuration that consists of one label $\P$ and $\Delta - 1$ labels $\O$. We will use $[\L_1\L_2]$ to represent the fact that we have a choice between the two labels $\L_1$ and $\L_2$. For example, we use $\M\s[\P\O]$ to represent the edge configuration that on one side has label $\M$ and on the other side has either label $\P$ or label $\O$. In other words, $\M\s[\P\O]$ is a condensed way to represent both configurations $\M\s\P$ and $\M\s\O$. We will call $[\L_1 \L_2]$ a \emph{disjunction}. Note that a configuration that contains a disjunction is, per se, a collection of configurations. For simplicity, we will refer to a configuration containing a disjunction as a \emph{condensed configuration}. We will say that a configuration $C$ is \emph{contained} in a condensed configuration $C'$ if there is a choice in $C'$ that results in $C$. For example, we say that the configuration $\M\s\P$ is in the condensed configuration $\M\s[\P\O]$.

\paragraph{Example: MIS.} As an example, let us see how one can encode the maximal independent set problem in the round elimination framework. For this problem, we define $\Sigma=\{\M,\P,\O\}$ (note that it has been shown that we need $3$ labels in order to encode MIS in this framework~\cite{binaryLCL}). The node constraint $\nodeconst$ consists of two configurations: one used by the nodes in the MIS and the other used by nodes not in the MIS. Intuitively, nodes in the MIS will output $\M$ on each incident edge, indicating that they are in the MIS. Nodes not in the MIS will use label $\P$ to \emph{point} to exactly one neighbor in the MIS, ensuring maximality, and label $\O$ (as ``other'') on the other $\Delta-1$ incident edges. The edge constraint $\edgeconst$ will take care of ensuring the independence: in fact, the configuration $\M\s\M$ is not allowed, meaning that two nodes in the MIS cannot be neighbors. Also, the edge constraint must ensure that the nodes not in the MIS correctly point at nodes in the MIS ($\M\s\P\in\edgeconst$), and that they cannot point to each other ($\P\s\P\notin\edgeconst$) or to nodes not in the MIS ($\O\s\P\notin\edgeconst$). Moreover, since two nodes not in the MIS can be neighbors, we allow the configuration $\O\s\O$ in \edgeconst. More precisely, the node and edge constraint for the MIS problem are the following.

\begin{equation*}
	\begin{aligned}
		\begin{split}
			\nodeconst\text{:} \\ 
			& \quad\M^{\Delta} \\
			& \quad\P\s\O^{\Delta-1} 
		\end{split}
		\qquad
		\begin{split}
			\edgeconst\text{:} \\
			& \quad \M\s[\P\O] \\
			& \quad \O\s\O
		\end{split}
	\end{aligned}
\end{equation*}

\subsection{Round Elimination}\label{sec:re}

The technique that we use for showing our lower bounds is the so-called \emph{round elimination} technique. More precisely, we will use the result of \cite[Theorem 4.3]{Brandt2019}, which informally says how to use, in an automatic way, the round elimination technique on locally checkable problems, if the input graph is a $\Delta$-regular high-girth graph.  In particular, given a locally checkable problem $\Pi$ with time complexity $T$, this theorem says how to construct a problem $\Pi''$ having time complexity exactly $\max\{T-1,0\}$. The procedure for obtaining $\Pi''$ actually goes through an intermediate step: from $\Pi$ with time complexity $T$ it first constructs a problem $\Pi'$, then from $\Pi'$ it constructs problem $\Pi''$ with time complexity $T-1$. Since $\Pi'$ is uniquely defined by $\Pi$, and $\Pi''$ is uniquely defined by $\Pi'$, we define a function $\re(\cdot)$ that with $\Pi$ as input returns $\Pi'$, and a function $\rere(\cdot)$ that with $\Pi'$ as input returns $\Pi''$. Hence, we have that $\Pi''=\rere(\re(\Pi))$.

Let $\Pi$ be the problem characterized by the alphabet set $\Sigma_{\Pi}$, by the node constraint $\nodeconst_{\Pi}$, and by the edge constraint $\edgeconst_{\Pi}$. In the following, we will formally show how to construct problems $\Pi'=\re(\Pi)$ and $\Pi''=\rere(\re(\Pi))$.

\paragraph{Problem $\Pi'=\re(\Pi)$.} In order to define $\Pi'$, we must define $\Sigma_{\Pi'}$, $\nodeconst_{\Pi'}$, and $\edgeconst_{\Pi'}$.

\begin{itemize}

	\item $\edgeconst_{\Pi'}$: the edge constraint of problem $\Pi'$ is defined as follows. Consider the collection $\mathcal{A}$ of all configurations $\A_1 \s \A_2$ where $\A_1,\A_2\in 2^{\Sigma_{\Pi}}\setminus\{\{\}\}$, such that, \emph{for all} $(\a_1, \a_2) \in \A_1 \times \A_2$, it holds that $\a_1 \s \a_2\in \edgeconst_{\Pi}$. Before defining the edge constraint of $\Pi'$ we must introduce the notion of \emph{non-maximality}. Informally, we say that a configuration $\A_1 \s \A_2$ is non-maximal if there exists another configuration $\A'_1 \s \A'_2$ that is obtained starting from $\A_1 \s \A_2$ and adding at least one element to at least one of $\A_1$ and $\A_2$. More formally, $\A_1 \s \A_2 \in \mathcal{A}$ is non-maximal if there exists another configuration $\A'_1 \s \A'_2 \in \mathcal{A}$ such that, $\forall i \in \{1,2\}, \A_i \subseteq \A'_i$ and $\exists i \in \{1,2\} : \A_i \subsetneq \A'_i$. Let $\mathcal{S}$ be the set of all non-maximal configurations in $\mathcal{A}$. The edge constraint of $\Pi'$ is then defined as $\edgeconst_{\Pi'}=\mathcal{A}\setminus\mathcal{S}$.
	
	\item $\Sigma_{\Pi'}$: the alphabet set $\Sigma_{\Pi'}$ is defined as a subset of the set containing all non-empty subsets of $\Sigma_{\Pi}$, that is $\Sigma_{\Pi'} \subseteq 2^{\Sigma_{\Pi}}\setminus\{\{\}\}$. In particular, it contains all and only the sets that appear at least once in the configurations of $\edgeconst_{\Pi'}$.
	
	\item $\nodeconst_{\Pi'}$: the node constraint of problem $\Pi'$ is defined as the collection of all configurations $\B_1\s \B_2\s\ldots\s \B_\Delta$ such that, $\forall i\in\{1,\dotsc,\Delta\}$, $\B_i \in \Sigma_{\Pi'}$, and it holds that \emph{there exists} $(\b_1,\dotsc,\b_\Delta) \in \B_1 \times \dotsc \times \B_\Delta$ satisfying that $\b_1\s \b_2\s\dotsc\s \b_\Delta \in \nodeconst_{\Pi}$.
	
\end{itemize}

\paragraph{Problem $\Pi''=\rere(\re(\Pi))$.} In order to define $\Pi''$, we must define $\Sigma_{\Pi''}$, $\nodeconst_{\Pi''}$, and $\edgeconst_{\Pi''}$.

\begin{itemize}
	\item $\nodeconst_{\Pi''}$: the node constraint of problem $\Pi''$ is defined as follows. Consider the collection $\mathcal{B}$ of all configurations $\B_1\s \B_2\s\ldots\s \B_\Delta$ where $\B_1, \B_2,\ldots, \B_\Delta\in 2^{\Sigma_{\Pi'}}\setminus\{\{\}\}$, such that, \emph{for all} $(\b_1,\dotsc,\b_\Delta) \in \B_1 \times \dotsc \times \B_\Delta$, it holds that $\b_1,\dotsc,\b_\Delta\in \nodeconst_{\Pi'}$. Let $\mathcal{S}$ be the set of all non-maximal configurations in $\mathcal{B}$. The node constraint of $\Pi''$ is then defined as $\nodeconst_{\Pi''}=\mathcal{B}\setminus\mathcal{S}$.
	
	\item $\Sigma_{\Pi''}$: the alphabet set $\Sigma_{\Pi''}$ is defined as a subset of the set containing all non-empty subsets of $\Sigma_{\Pi'}$, that is $\Sigma_{\Pi''}=\subseteq 2^{\Sigma_{\Pi'}}\setminus\{\{\}\}$. In particular, it contains all and only the sets that appear at least once in the configurations of $\edgeconst_{\Pi''}$.
	
	\item $\edgeconst_{\Pi''}$: the edge constraint of problem $\Pi''$ is defined as the collection of all configurations $\A_1\s \A_2$ such that, $\forall i\in\{1,2\}$, $\A_i$ appears in at least one configuration in $\nodeconst_{\Pi''}$, and it holds that \emph{there exists} $(\a_1,\a_2) \in \A_1 \times \A_2$ satisfying that $\a_1\s \a_2 \in \edgeconst_{\Pi'}$.
	
\end{itemize}
Note that there is a simple method for computing $\nodeconst_{\Pi'}$: take the node constraint of $\Pi$ and replace each label $\y$ in each configuration by the disjunction of all label sets in $\Sigma_{\Pi'}$ that contain $\y$. In a similar manner, we can also compute $\edgeconst_{\Pi''}$.

For more concrete examples of constructions of $\re(\Pi)$ and $\rere(\re(\Pi))$ where $\Pi$ is, for example, the sinkless orientation problem, we refer the reader to the tutorial of the round eliminator tool~\cite{Olivetti2019}.

The result of Brandt~\cite[Theorem 4.3]{Brandt2019} proves a useful relation between $\Pi$ and $\rere(\re(\Pi))$.

\begin{theorem}[\cite{Brandt2019}, rephrased]\label{thm:sebastien}
	Consider a class $\mathcal G$ of graphs with girth at least $2T+2$, and some locally checkable problem $\Pi$. 
	Then, there exists an algorithm that solves problem $\Pi$ on $\mathcal G$ in $T$ rounds if and only if there exists an algorithm that solves problem $\rere(\re(\Pi))$ in $\max\{T-1,0\}$ rounds.
\end{theorem}
For technical reasons, the above theorem holds for the port numbering model, and that is why we first prove our lower bounds for the port numbering model and then we lift these results to the \LOCAL model. Also, Theorem~\ref{thm:sebastien} holds even if nodes are provided with some input, if this input satisfies a property called \emph{$t$-independence}. In this work we assume that nodes are provided with a $\Delta$-edge coloring, and that each node knows the color of its incident edges. This input satisfies the property required to apply Theorem ~\ref{thm:sebastien}.

\paragraph{Node/Edge Diagram: Relation Between Labels.} 
Let $\Pi$ be a locally checkable problem, characterized by $\Sigma_{\Pi}$, $\nodeconst_{\Pi}$, and $\edgeconst_{\Pi}$, and let $\A,\B$ be two labels in $\Sigma_{\Pi}$. We now define a relation between labels, which will be useful later when we will compute $\re(\Pi)$ or $\rere(\Pi)$ for our problems of interest. We say that label $\A$ is \emph{at least as strong as} label $\B$ according to $\edgeconst_\Pi$, and equivalently $\B$ is \emph{at least as weak as} label $\A$ according to $\edgeconst_\Pi$, if for each edge configuration $C\in\edgeconst_\Pi$ containing $\B$, the operation of replacing one occurrence of $\B$ in $C$ by $\A$ results in a configuration that is also in $\edgeconst_\Pi$. Also, we say that $\A$ is \emph{stronger than} $\B$, and equivalently $\B$ is \emph{weaker than} $\A$, if $\A$ is at least as strong as $\B$ and $\B$ is not at least as strong as $\A$. The relations between labels of a problem $\Pi$ according to the edge constraint are illustrated by the \emph{edge diagram} of $\Pi$. Such a diagram is a directed graph where nodes are labels in $\Sigma_{\Pi}$ and there is a directed edge from label $\A$ to label $\B$ if $\B\neq \A$, $\B$ is at least as strong as $\A$, and there does not exist a label $\Z\in\Sigma_{\Pi}$ such that $\B$ is stronger than $\Z$ and $\Z$ is stronger than $\A$ (that is, the diagram illustrates  all those strength relations that cannot be decomposed into ``smaller'' strength relations). Analogously, we define the above relations \emph{according to the node constraint}. Similarly, we illustrate with the \emph{node diagram} the strength relations of labels according to the node constraint. For example, consider the MIS problem defined in Section~\ref{subsec:problems}. The edge diagram of the MIS problem is shown in Figure \ref{fig:mis}.

\begin{figure}[h]
	\centering
	\includegraphics[width=0.13\textwidth]{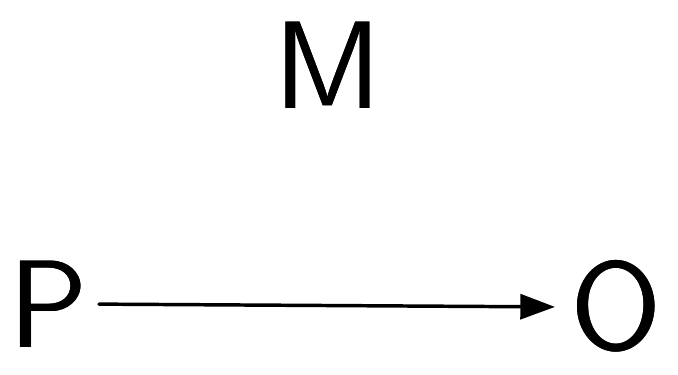}
	\caption{Relation between the labels of the MIS problem according to the edge constraint: $\O$ is stronger than $\P$, and there is no relation between labels $\M$ and $\P$, and between $\M$ and $\O$.}
	\label{fig:mis}
\end{figure}

\paragraph{Right-closed Sets.}
Consider a set $S = \{ \A_1, \dots, \A_p \} \subseteq \Sigma_\Pi$. We say that $S$ is \emph{right-closed} if and only if, for each $\A_i\in S$ it holds that also all successors of $\A_i$ are in $S$. Balliu, Brandt, and Olivetti~\cite{balliurules} proved the following observation about the relation of right-closed sets and the definitions of $\re(\cdot)$ and $\rere(\cdot)$.

\begin{observation}[\cite{balliurules}, Observation 8 in the arXiv version]\label{obs:rcs}
	Consider an arbitrary collection of labels $\A_1, \dots, \A_p \in \Sigma_\Pi$.
	If $\{ \A_1, \dots, \A_p \} \in \Sigma_{\re(\Pi)}$, then the set $\{ \A_1, \dots, \A_p \}$ is right-closed (w.r.t.\ $\edgeconst_\Pi$).
	If $\{ \A_1, \dots, \A_p \} \in \Sigma_{\rere(\Pi)}$, then the set $\{ \A_1, \dots, \A_p \}$ is right-closed (w.r.t.\ $\nodeconst_\Pi$).
\end{observation}

\subsection{Roadmap}
We will start by proving a lower bound for computing $k$-outdegree dominating sets in the deterministic port numbering model (Section \ref{sec:pn}). For this purpose, we will define a parameterized family of problems $\mathcal{P}$, and we will show that these problems can be solved in constant time given a $k$-outdegree dominating set, which means that if we can show a lower bound for some of the problems in $\mathcal{P}$, then we would obtain the same asymptotic lower bound for $k$-outdegree dominating sets.
Then, we will relate different problems of $\mathcal{P}$, by showing that some of them are strictly easier than others. We will apply this reasoning multiple times in order to build a long sequence of problems, such that each one is strictly easier than the previous one, and also the last one is not $0$-rounds solvable. In this way we will obtain that the length of the chain is a lower bound for $k$-outdegree dominating sets.

In order to relate two different problems $P,P'$ of $\mathcal{P}$ we use the round elimination technique. Informally this technique is usually applied in the following way: 
\begin{enumerate}
	\item start from a problem $P$ in our family of problems $\mathcal{P}$,
	\item apply the round elimination technique and obtain a new problem $\Pi$ that might not be a problem in $\mathcal{P}$,
	\item prove that $\Pi$ is not easier than some problem $P'\in \mathcal{P}$,
	\item this implies that $P'$ is at least one round easier than $P$.
\end{enumerate}
For the problem family that we define, this approach does not seem to work: by applying the round elimination technique we will obtain some problem that is very similar to the problems of the family, but it will also contain some additional allowed configuration. This additional configuration will prevent us from proving that the obtained problem is not easier than some problem of our family. Here lies the novelty of our approach: in order to get rid of such a configuration we will exploit a specific input given to the nodes (a $\Delta$-edge coloring) in such a way that nodes, in $0$ rounds, can convert their output to a different one that does \emph{not} use the additional configuration, and hence solve a problem that is actually in the family.

Finally, in Section \ref{sec:local} we will use standard techniques to convert our deterministic port numbering lower bound into a lower bound for the \LOCAL model.

\section{Lower bound for the port numbering model}\label{sec:pn}
In this section we will define a family of problems $\Pi_{\Delta}(a,x)$ that we will use to prove a lower bound for $k$-outdegree dominating sets.  The parameters $a$ and $x$ satisfy $0 \le a,x \le \Delta$, and intuitively by increasing $x$, or decreasing $a$, we will get easier problems. We will then relate the problems of the family, by showing that some of them are at least one round harder than others. In particular, we will prove that for $a$ large enough and $x$ small enough, $\Pi_{\Delta}(a,x)$ requires at least one round more than $\Pi_{\Delta}(\lfloor(a-2x-1)/2\rfloor,x+1)$. We will finally prove that, even for very relaxed problems of the family, that is those having small values of $a$ and large values of $x$, $\Pi_{\Delta}(a,x)$ is not $0$ rounds solvable. At the end we will combine all these results to show that, for a wide range of values of $k$, there exists a sequence of problems $\Pi_0 \rightarrow \Pi_1 \ldots \rightarrow \Pi_T$, where $\Pi_0$ can be solved in $1$ round given a solution for $k$-outdegree dominating sets, $\Pi_{i+1}$ is at least one round easier than $\Pi_i$, and $\Pi_T$ is not $0$ rounds solvable, implying a lower bound of $T$ rounds for $k$-outdegree dominating sets.

\subsection{The Problem Family}
On a high level, our problem family contains the following problems. Consider the following task: find an independent set of nodes (type-$1$ nodes), such that nodes that are not in the independent set either have a neighbor in the independent set (type-$2$ nodes), or they own at least $a$ of their incident edges (type-$3$ nodes; an edge can be owned by at most one of its endpoints). We define a problem that is a slight relaxation of this task. Type-$2$ nodes may not have any neighbor in the independent set, but in that case they must have at least one neighbor of type $3$ connected to them through a non-owned edge. Moreover, we allow type-$1$ nodes to violate the independence requirement: these nodes are allowed to be neighbors, but the edges between them should be properly oriented, such that at most $x$ edges are outgoing for each node in the set. See Figure~\ref{fig:problemFamily} for an example.
We will now define the problem $\Pi_{\Delta}(a,x)$ formally, by providing its label set and its node and edge constraints.

\begin{figure}[t]
	\centering
	\includegraphics[width=0.54\textwidth]{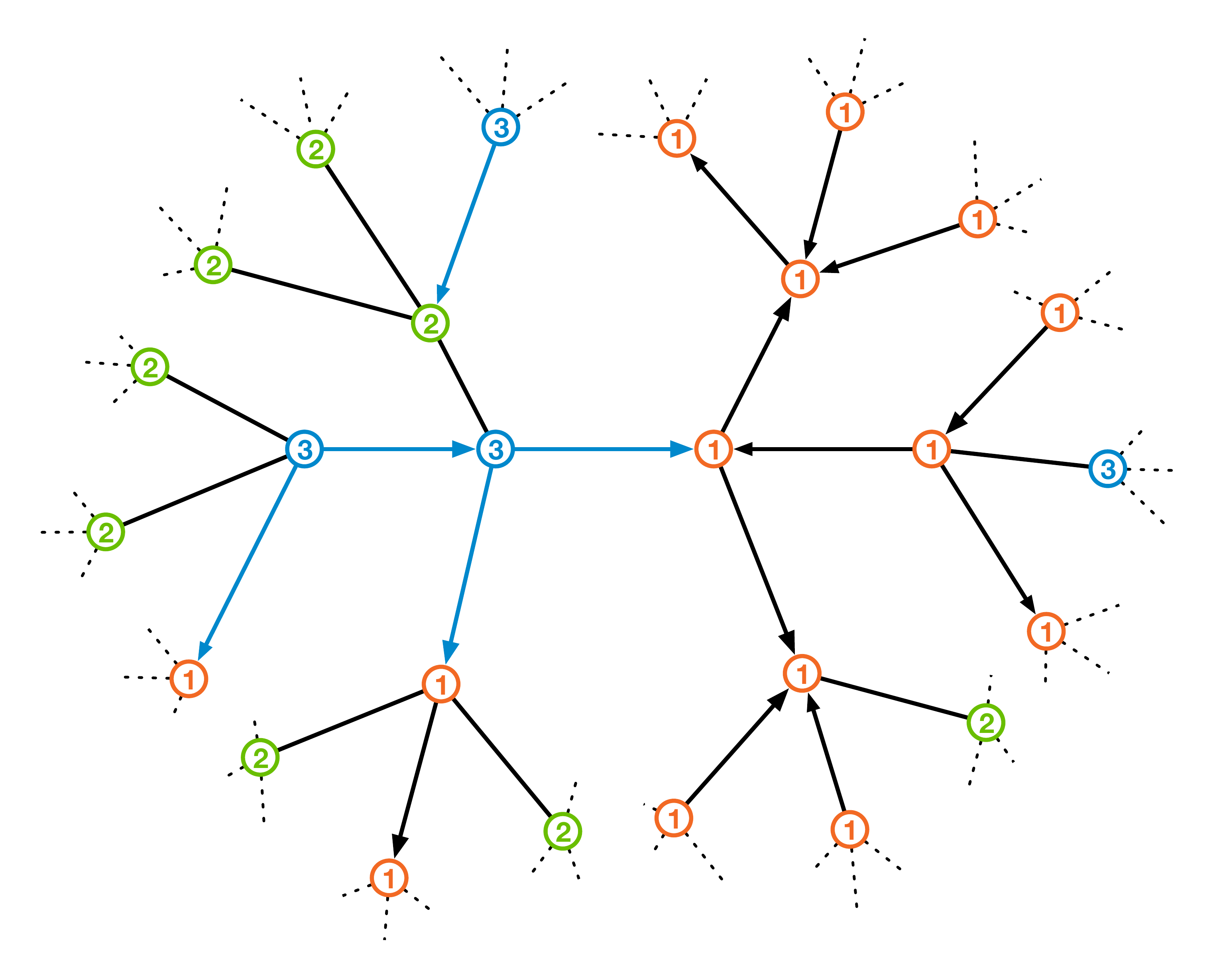}
	\caption{Example of a problem in our family, where $a=2$ and $x=2$. Directed blue edges that go from a type-$3$ node $u$ to a neighbor $v$ of $u$ are owned by node $u$.}
	\label{fig:problemFamily}
\end{figure}

\paragraph{Labels.} For all problems of the family, the label set $\Sigma_{\Delta}(a,x)$ is $\{\M, \P, \O, \A, \X\}$. Informally, these labels can be interpreted as follows:
\begin{itemize}
	\item The label $\M$ is used by nodes that are in the dominating set.
	\item Nodes not in the dominating set can use label $\P$ to point to a neighbor in the dominating set, and the label $\O$ is used to label the other incident edges.
	\item The label $\A$ is used by nodes to mark \emph{owned} edges, while label $\X$ is used to mark other incident edges. 
	\item Also, nodes of the dominating set can use label $\X$ to mark edges that connect them to neighbors that are also in the dominating set.
\end{itemize}

\paragraph{Node Constraint.} We now define the node constraint $\nodeconst_{\Delta}(a,x)$ of $\Pi_{\Delta}(a,x)$ by listing its allowed configurations. $\nodeconst_{\Delta}(a,x)$ contains the following:
\begin{itemize}
	\item  $\M^{\Delta-x} \X^x$. This configuration is used by nodes that are part of the dominating set.
	\item  $\A^{a} \s \X^{\Delta - a}$. This configuration is used by some nodes that are not in the dominating set in order to prove that they own at least $a$ edges.
	\item  $\P \s \O^{\Delta - 1}$. This configuration is used by some nodes $u$ that are not in the dominating set in order to prove that they have at least one neighbor in the dominating set (or a neighbor $v$ owning $a$ edges, but not the one connecting $v$ to $u$).
\end{itemize}

\paragraph{Edge Constraint.}
 We now define the edge constraint $\edgeconst_{\Delta}(a,x)$ of $\Pi_{\Delta}(a,x)$ by listing its allowed configurations. $\edgeconst_{\Delta}(a,x)$ contains the following:
\begin{itemize}
	\item $\M \s  [\P\A\O\X]$
	\item $\O  \s [\M\A\O\X]$
	\item $\P  \s  [\M\X]$
	\item $\A  \s  [\M\O\X]$
	\item $\X  \s  [\M\P\A\O\X]$
\end{itemize}

\begin{figure}[t]
	\centering
	\includegraphics[width=0.54\textwidth]{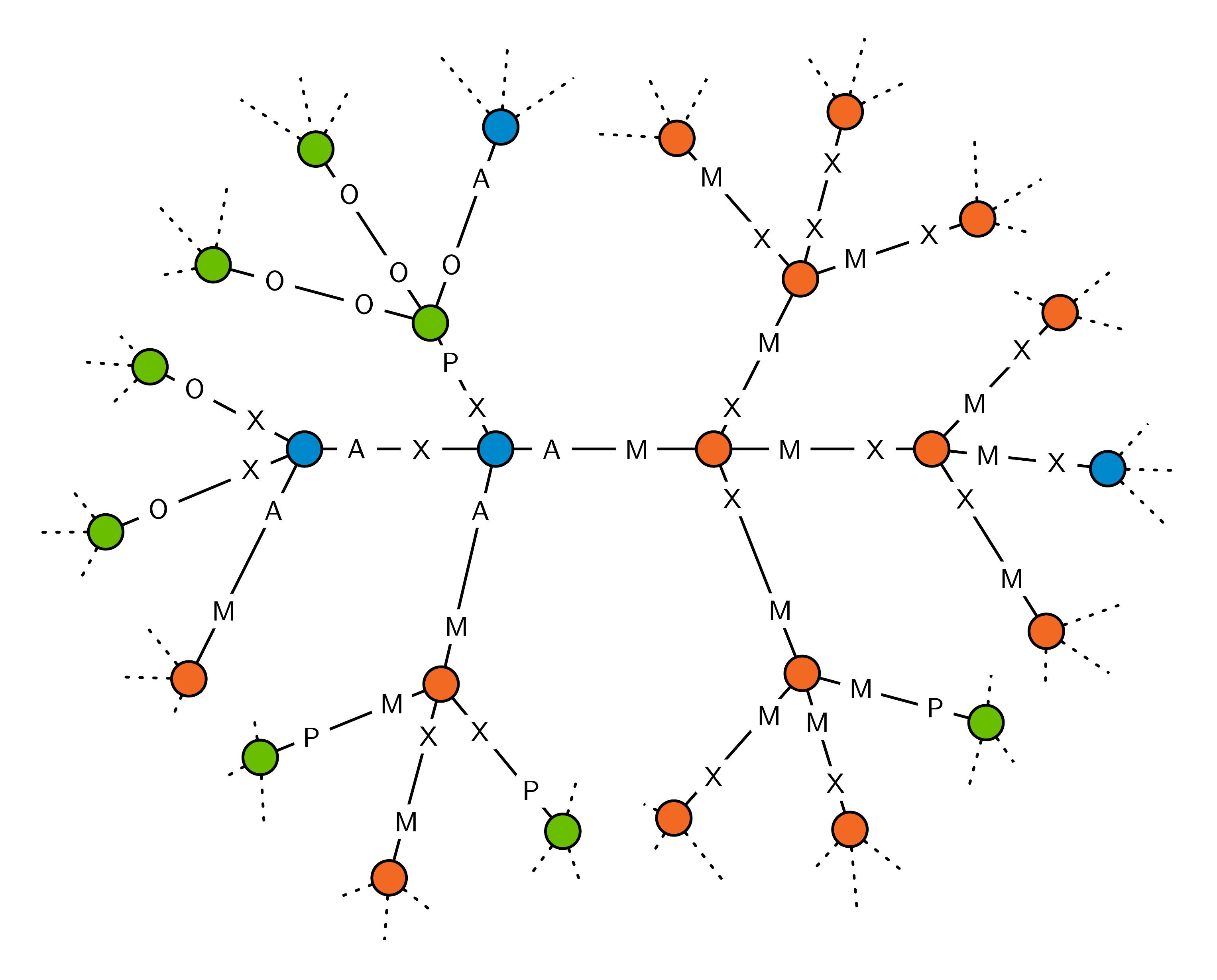}
	\caption{Example of an output labeling that solves a problem in our family, where $a=x=2$ and $\Delta=4$. Type-$1$, type-$2$, and type-$3$ nodes are represented by orange, green, and blue nodes, respectively.}
	\label{fig:problemFamily2}
\end{figure}

In other words, $\M$ is not compatible with $\M$, $\A$ is not compatible with $\A$, $\P$ is not compatible with $\P$, $\A$ or $\O$, while anything else is allowed. See Figure \ref{fig:problemFamily2} for an example.
The edge diagram of $\Pi_\Delta(a,x)$ is as given in Figure \ref{fig:firstdiagram}

\begin{figure}[ht]
	\centering
	\includegraphics[width=0.35\textwidth]{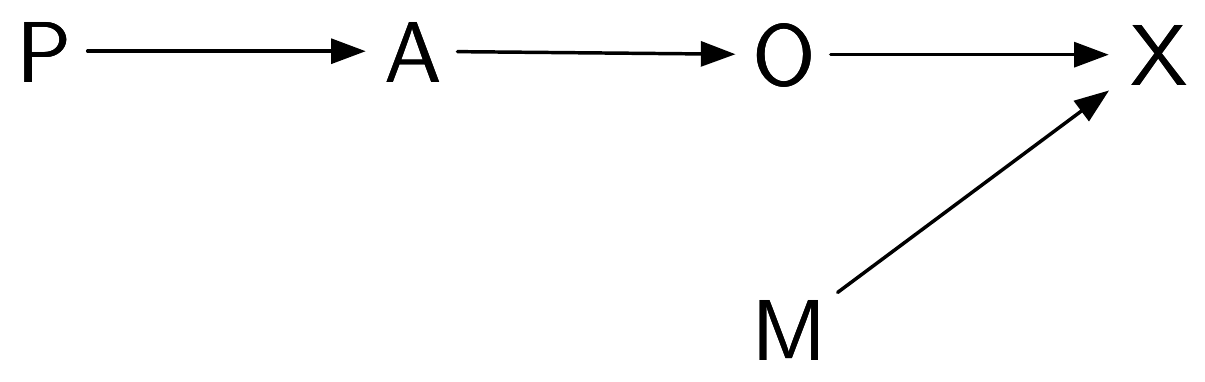}
	\caption{Relation between the labels of $\Pi_{\Delta}(a,x)$ according to the edge constraint.}
	\label{fig:firstdiagram}
\end{figure}

\subsection{Relation Between \texorpdfstring{\boldmath $k$}{k}-Outdegree Dominating Sets and the Problems of the Family}
We start by relating $k$-outdegree dominating sets with the problems of the family.
\begin{lemma}\label{lem:tofamily}
	Given a solution for the $k$-outdegree dominating set problem, we can solve $\Pi_{\Delta}(a,k)$ in $1$ round, for all $a$.
\end{lemma}
\begin{proof}
	We are given in input a solution to the $k$-outdegree dominating set problem, which means that nodes in the dominating set form a directed graph where every node has at most $k$ outgoing edges. Such nodes output $\X$ on the outgoing edges and $\M$ on the others. These nodes modify this labeling by writing $\X$ on some edges labeled $\M$ in order to obtain exactly $k$ edges labeled $\X$. Nodes not in the dominating set spend $1$ round to check which of the neighbors are in the dominating set, and they output $\P$ on one port connecting them to a dominating set neighbor and $\O$ on all the others. Clearly, the node constraint is satisfied. We now show that the edge constraint is satisfied as well.
	\begin{itemize}
		\item An edge that connects two nodes in the dominating set will never be labeled with $\M$ by both endpoints, since outgoing edges are marked $\X$. All the other labels used by the algorithm are $\X$, $\P$, and $\O$, which are all compatible with $\M$ and $\X$.
		\item If $\P$ is written on one side of an edge, then the other side is labeled with either $\M$ or $\X$, which are both compatible with $\P$.
		\item The label $\O$ cannot appear on the other side of an edge labeled $\P$, and all the other labels used by the algorithm, i.e., $\X$, $\M$, and $\O$, are all compatible with $\O$.
		\item The label $\X$ is compatible with everything.\hspace*{\fill}\qedhere
	\end{itemize}
\end{proof}

\subsection{Relation Between Problems of the Family}
In this section we prove that $\Pi_{\Delta}(\lfloor(a-2x-1)/2\rfloor,x+1)$ is at least as easy as $\rere(\re(\Pi_{\Delta}(a,x)))$, implying, by Theorem~\ref{thm:sebastien}, that $\Pi_{\Delta}(\lfloor(a-2x-1)/2\rfloor,x+1)$ is at least one round easier than $\Pi_{\Delta}(a,x)$, if $\Pi_{\Delta}(a,x)$ is not zero round solvable. We start by computing $\re(\Pi_{\Delta}(a,x))$ and use the obtained knowledge to prove that $\rere(\re(\Pi_{\Delta}(a,x)))$ is at least as hard as the following problem $\Pi^+_{\Delta}(a,x)$:
\begin{equation*}
\begin{aligned}
\begin{split}
\nodeconst^+_{\Delta}(a,x)\text{:} \\ 
& \quad\M^{\Delta-x-1} \X^{x+1} \\
& \quad\P\s\O^{\Delta-1} \\
& \quad\A^{a-x-1}\s\X^{\Delta-a+x+1}\\
& \quad\C^{\Delta-x}\s\X^{x} \\\\\\
\end{split}
\qquad
\begin{split}
\edgeconst^+_{\Delta}(a,x)\text{:} \\
& \quad \M\s[\P\A\C\O\X] \\
& \quad \O\s[\M\A\C\O\X] \\
& \quad \P\s[\M\X] \\
& \quad \A\s[\M\C\O\X] \\
& \quad \X\s[\M\P\A\C\O\X] \\
& \quad \C\s[\M\A\O\X]
\end{split}
\end{aligned}
\end{equation*}
Then, we will show how to exploit a given $\Delta$-edge coloring to convert any solution for such a problem into a solution for $\Pi_{\Delta}(\lfloor(a-2x-1)/2\rfloor,x+1)$. Intuitively, we will get rid of nodes labeled with the configuration $\C^{\Delta-x}\s\X^{x}$, by making them output the configuration $\A^{\lfloor(a-2x-1)/2\rfloor}\s\X^{\Delta-\lfloor(a-2x-1)/2\rfloor}$. In the process, also the output of nodes labeled with the configuration $\A^{a-x-1}\s\X^{\Delta-a+x+1}$ will be altered, by reducing their number of incident edges labeled $\A$ from $a-x-1$ to $\lfloor(a-2x-1)/2\rfloor$. Note that these changes in the output give a solution for $\Pi_{\Delta}(\lfloor(a-2x-1)/2\rfloor,x+1)$. We start by computing $\re(\Pi_{\Delta}(a,x))$.

\begin{lemma}\label{lem:computere}
	Assume that $x+2 \leq a \leq \Delta$.
	After renaming, the node constraint of $\re(\Pi_{\Delta}(a,x))$ is given by the configurations
	\begin{align*}
		& \quad[\zx\zo\za\zp]^{\Delta-x} [\xx\zx\xo\zo\xa\za\xp\zp]^{x} \\
		& \quad[\xp\zp]\s[\xo\zo\xa\za\xp\zp]^{\Delta-1} \\
		& \quad[\xa\za\xp\zp]^a [\xx\zx\xo\zo\xa\za\xp\zp]^{\Delta - a} 
	\end{align*}
	and the edge constraint of $\re(\Pi_{\Delta}(a,x))$ by the configurations
	\begin{align*}
		& \quad\xx\s\zp \\
		& \quad\xo\s\za \\
		& \quad\xa\s\zo \\
		& \quad\xp\s\zx
	\end{align*}
	In particular, the node diagram of $\re(\Pi_{\Delta}(a,x))$ is as given in Figure~\ref{fig:nodia}.
\end{lemma}
\begin{proof}

By Observation~\ref{obs:rcs}, the set of labels of $\re(\Pi_{\Delta}(a,x))$ is a subset of the set of right-closed sets of labels according to the diagram of $\Pi_\Delta(a,x)$, depicted in Figure \ref{fig:firstdiagram}. All possible right-closed sets are $\mathcal{S} = \{\mybox{\X},\mybox{\M\X},\mybox{\O\X},\mybox{\M\O\X},\mybox{\A\O\X},\mybox{\M\A\O\X},\mybox{\P\A\O\X},\mybox{\M\P\A\O\X}\}$, where $\mybox{L_1 \ldots L_k}$ represents the set of labels $\{L_1,\ldots,L_k\}$. In order to compute the edge constraint of $\re(\Pi_{\Delta}(a,x))$, we can pair each $S \in \mathcal{S}$ with the unique largest set of labels $S' \in \mathcal{S}$ satisfying that for any choice $(s,s') \in (S,S')$, the configuration $s \s s'$ is in $\edgeconst_{\Delta}(a,x)$. By performing such operation and removing all symmetric configurations, we obtain the following:
\begin{equation*}
\begin{aligned}
	 &\mybox{\X} \s \mybox{\M\P\A\O\X}\\
	 &\mybox{\M\X} \s \mybox{\P\A\O\X}\\
     &\mybox{\O\X} \s \mybox{\M\A\O\X}\\
	 &\mybox{\M\O\X} \s \mybox{\A\O\X}\\
\end{aligned}
\end{equation*}
Note that all obtained configurations are also maximal. Also, notice that, by renaming labels as in the following mapping, we obtain the required edge constraint.
\begin{align*}
\mybox{\X} &\mapsto \xx \\
\mybox{\M\X} &\mapsto \zx \\
\mybox{\O\X} &\mapsto \xo \\
\mybox{\M\O\X} &\mapsto \zo \\
\mybox{\A\O\X} &\mapsto \xa \\
\mybox{\M\A\O\X} &\mapsto \za \\
\mybox{\P\A\O\X} &\mapsto \xp \\
\mybox{\M\P\A\O\X} &\mapsto \zp \\
\end{align*}
Using the method explained in Section~\ref{sec:re}, we can compute the node constraint of $\re(\Pi_{\Delta}(a,x))$ by simply taking the node constraint of $\Pi_{\Delta}(a,x)$ and replacing each label $\y$ in each configuration by the disjunction of all label sets that appear in the edge constraint of $\re(\Pi_{\Delta}(a,x))$ and contain $\y$. Note that the original label $\M$ is contained in the sets represented by the new labels $\zx\zo\za\zp$, the original label $\X$ is contained all the new labels, $\A$ is contained in $\xa\za\xp\zp$, the original label $\O$ is contained in the new labels $\xo\zo\xa\za\xp\zp$, and the original label $\P$ is contained in the new labels $\xp\zp$. Hence, by applying the replacing method on the node constraints of  $\Pi_{\Delta}(a,x)$, we obtain the same node constraint as in the claim.
\end{proof}

\begin{figure}[t]
	\centering
	\includegraphics[width=0.35\textwidth]{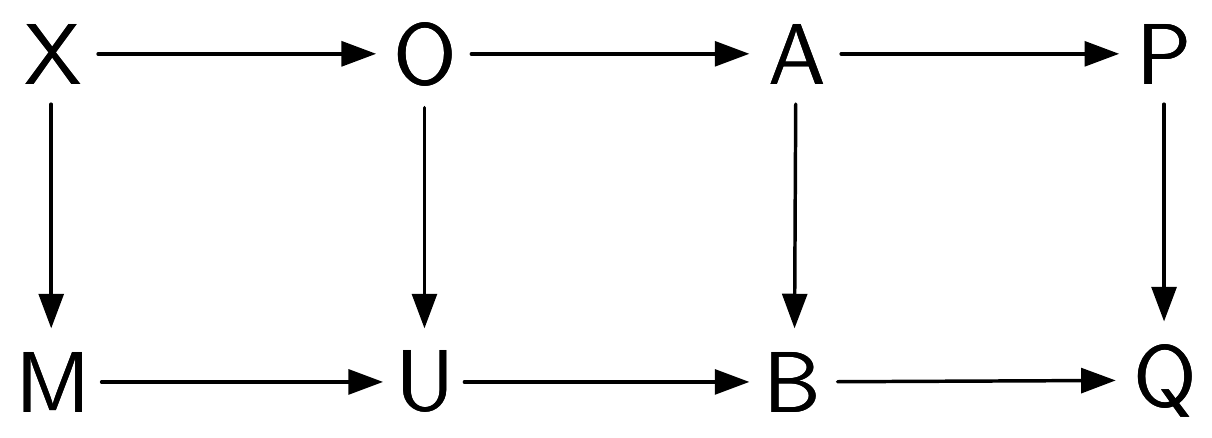}
	\caption{Relation between the labels of $\re(\Pi_{\Delta}(a,x))$ according to the node constraint.}
	\label{fig:nodia}
\end{figure}

We now compute $\rere(\re(\Pi_{\Delta}(a,x)))$. We will need the notion of a relaxation of a node configuration as given in~\cite{balliurules}.
\begin{definition}\label{def:relax}
	Let $\Y_1, \Y_2, \dots, \Y_\Delta$ be sets of labels.
	A relaxation of a node configuration $\Y_1\s\Y_2\s\dots\s\Y_\Delta$ is a node configuration $\Z_1\s\Z_2\s\dots\s\Z_\Delta$ such that there exists a permutation $\rho \colon \{1, \dots, \Delta\} \to \{1, \dots, \Delta\}$ such that, for each $1 \leq i \leq \Delta$, we have $\Y_i \subseteq Z_{\rho(i)}$.
	If $\Z_1\s\Z_2\s\dots\s\Z_\Delta$ is a relaxation of $\Y_1\s\Y_2\s\dots\s\Y_\Delta$, then we say that $\Y_1\s\Y_2\s\dots\s\Y_\Delta$ \emph{can be relaxed to} $\Z_1\s\Z_2\s\dots\s\Z_\Delta$.
\end{definition}

\begin{lemma}\label{lem:speedup}
	If the problem $\Pi_{\Delta}(a,x)$ has time complexity $T$, then $\Pi^+_{\Delta}(a,x)$ has time complexity $\max\{T-1,0\}$, for all $a, x$ satisfying $x+2 \leq a \leq \Delta$.
\end{lemma}
\begin{proof}
	Let $a$ and $x$ satisfy  $x+2 \leq a \leq \Delta$.
	We prove the lemma by showing that we can solve $\Pi^+_{\Delta}(a,x)$ in $0$ rounds given an arbitrary solution for $\rere(\re(\Pi_{\Delta}(a,x)))$ (which by definition is precisely one round easier than $\Pi_{\Delta}(a,x)$, unless $\Pi_{\Delta}(a,x)$ is already zero round solvable).
	However, we will not fully compute $\rere(\re(\Pi_{\Delta}(a,x)))$ for this purpose.
	Instead, we will implicitly show that the problem $\Pi_{\operatorname{rel}}$ given by the following node and edge constraints can be solved in $0$ rounds, given a solution for $\rere(\re(\Pi_{\Delta}(a,x)))$. 
	\begin{equation*}
		\begin{aligned}
			\begin{split}
				& \quad\mybox{\zx\zo\za\zp}^{\Delta-x-1} \mybox{\xx\zx\xo\zo\xa\za\xp\zp}^{x+1} \\
				& \quad\mybox{\xp\zp}\s\mybox{\xo\zo\xa\za\xp\zp}^{\Delta-1} \\
				& \quad\mybox{\xa\za\xp\zp}^{a-x-1}\s\mybox{\xx\zx\xo\zo\xa\za\xp\zp}^{\Delta-a+x+1}\\
				& \quad\mybox{\zo\za\xp\zp}^{\Delta-x}\s\mybox{\xx\zx\xo\zo\xa\za\xp\zp}^{x} \\\\\\
			\end{split}
			\qquad
			\begin{split}
				& \quad \mybox{\zx\zo\za\zp}\s[\mybox{\xp\zp}\mybox{\xa\za\xp\zp}\mybox{\zo\za\xp\zp}\mybox{\xo\zo\xa\za\xp\zp}\mybox{\xx\zx\xo\zo\xa\za\xp\zp}] \\
				& \quad \mybox{\xo\zo\xa\za\xp\zp}\s[\mybox{\zx\zo\za\zp}\mybox{\xa\za\xp\zp}\mybox{\zo\za\xp\zp}\mybox{\xo\zo\xa\za\xp\zp}\mybox{\xx\zx\xo\zo\xa\za\xp\zp}] \\
				& \quad \mybox{\xp\zp}\s[\mybox{\zx\zo\za\zp}\mybox{\xx\zx\xo\zo\xa\za\xp\zp}] \\
				& \quad \mybox{\xa\za\xp\zp}\s[\mybox{\zx\zo\za\zp}\mybox{\zo\za\xp\zp}\mybox{\xo\zo\xa\za\xp\zp}\mybox{\xx\zx\xo\zo\xa\za\xp\zp}] \\
				& \quad \mybox{\xx\zx\xo\zo\xa\za\xp\zp}\s[\mybox{\zx\zo\za\zp}\mybox{\xp\zp}\mybox{\xa\za\xp\zp}\mybox{\zo\za\xp\zp}\mybox{\xo\zo\xa\za\xp\zp}\mybox{\xx\zx\xo\zo\xa\za\xp\zp}] \\
				& \quad \mybox{\zo\za\xp\zp}\s[\mybox{\zx\zo\za\zp}\mybox{\xa\za\xp\zp}\mybox{\xo\zo\xa\za\xp\zp}\mybox{\xx\zx\xo\zo\xa\za\xp\zp}]
			\end{split}
		\end{aligned}
	\end{equation*}
	We will start by showing that any node configuration of $\rere(\re(\Pi_{\Delta}(a,x)))$ can be relaxed to some node configuration of $\Pi_{\operatorname{rel}}$.
	For a contradiction, suppose that there is a node configuration $\Y_1\s\Y_2\s\dots\s\Y_\Delta$ of $\rere(\re(\Pi_{\Delta}(a,x)))$ that cannot be relaxed to any node configuration of $\Pi_{\operatorname{rel}}$.
	Note that, by Observation~\ref{obs:rcs}, the $Y_i$ are subsets of the label set $\mybox{\xx\zx\xo\zo\xa\za\xp\zp}$ that are right-closed (w.r.t.\ the diagram given in Figure~\ref{fig:nodia}).
	We now use the fact that $\Y_1\s\Y_2\s\dots\s\Y_\Delta$ cannot be relaxed to any node configuration of $\Pi_{\operatorname{rel}}$ to collect information about the $\Y_i$.
	
	At least $x+2$ of the $Y_i$ must contain the label $\xp$ as otherwise at least $\Delta - x - 1$ of the $Y_i$ would be a subset of the set $\mybox{\zx\zo\za\zp}$ (due to the right-closedness of the $Y_i$), which in turn would imply that $\Y_1\s\Y_2\s\dots\s\Y_\Delta$ can be relaxed to the node configuration $\mybox{\zx\zo\za\zp}^{\Delta-x-1} \mybox{\xx\zx\xo\zo\xa\za\xp\zp}^{x+1}$.
	Using an analogous argument, we see that at least $\Delta - a + x + 2$ of the $Y_i$ must contain the label $\zo$ as otherwise $\Y_1\s\Y_2\s\dots\s\Y_\Delta$ could be relaxed to $\mybox{\xa\za\xp\zp}^{a-x-1}\s\mybox{\xx\zx\xo\zo\xa\za\xp\zp}^{\Delta-a+x+1}$.

	Now if one of the $Y_i$ contained the label $\zx$, it would be possible to select labels $y_1, y_2, \dots, y_\Delta$ from $\Y_1, \Y_2, \dots, \Y_\Delta$, respectively, such that the configuration $\y_1\s\y_2\s\dots\s\y_\Delta$ would contain at least one $\zx$, at least $(x + 2) - 1 = x + 1$ many $\xp$, and at least $(\Delta - a + x + 2) - 1 - (x + 1) = \Delta - a$ many $\zo$.
	However, by Lemma~\ref{lem:computere}, no such configuration is contained in the node constraint of $\re(\Pi_{\Delta}(a,x))$, which implies, by the definition of $\rere(\cdot)$, that $\Y_1\s\Y_2\s\dots\s\Y_\Delta$ is not a node configuration of $\rere(\re(\Pi_{\Delta}(a,x)))$, yielding a contradiction.
	Hence, none of the $Y_i$ contains the label $\zx$, and by the right-closedness of the $Y_i$ it follows that each $Y_i$ is a subset of $\mybox{\xo\zo\xa\za\xp\zp}$.
	This has two consequences.

	First, each of the $Y_i$ contains the label $\za$ as otherwise the right-closedness of the $Y_i$ would imply that $\Y_1\s\Y_2\s\dots\s\Y_\Delta$ could be relaxed to $\mybox{\xp\zp}\s\mybox{\xo\zo\xa\za\xp\zp}^{\Delta-1}$.
	Second, at least $x+1$ of the $Y_i$ contain the label $\xa$ as otherwise $\Y_1\s\Y_2\s\dots\s\Y_\Delta$ could be relaxed to $\mybox{\zo\za\xp\zp}^{\Delta-x}\s\mybox{\xx\zx\xo\zo\xa\za\xp\zp}^{x}$.

	Now we have all the ingredients to obtain the desired contradiction.
	The latter two observations together with the earlier obtained fact that at least $\Delta - a + x + 2$ of the $Y_i$ contain the label $\zo$ imply that we can select labels $y_1, y_2, \dots, y_\Delta$ from $\Y_1, \Y_2, \dots, \Y_\Delta$, respectively, such that the configuration $\y_1\s\y_2\s\dots\s\y_\Delta$ would contain $x + 1$ many $\xa$ and $(\Delta - a + x + 2) - (x + 1) = \Delta - a + 1$ many $\zo$, and the remaining labels would be $\za$.
	However, such a configuration does not exist in the node constraint of of $\re(\Pi_{\Delta}(a,x))$, which again implies that $\Y_1\s\Y_2\s\dots\s\Y_\Delta$ is not a node configuration of $\rere(\re(\Pi_{\Delta}(a,x)))$, yielding a contradiction.
	Thus, any node configuration of $\rere(\re(\Pi_{\Delta}(a,x)))$ can be relaxed to some node configuration of $\Pi_{\operatorname{rel}}$.

	Using the method explained in Section~\ref{sec:re}, we can compute the edge constraint of $\rere(\re(\Pi_{\Delta}(a,x)))$ by simply taking the edge constraint of $\re(\Pi_{\Delta}(a,x))$ and replacing each label $\y$ in each configuration by the disjunction of all label sets that appear in the node constraint of $\rere(\re(\Pi_{\Delta}(a,x)))$ and contain $\y$.
	Due to Definition~\ref{def:relax} and the fact that the edge constraint of $\Pi_{\operatorname{rel}}$ is obtained by precisely the aforementioned method (with the only difference that the label sets from the node configurations of $\Pi_{\operatorname{rel}}$ are used instead of those of $\rere(\re(\Pi_{\Delta}(a,x)))$, and that the same configurations are represented by using different condensed configurations), we see that there is a simple $0$-round algorithm that computes a correct solution for $\Pi_{\operatorname{rel}}$ given a solution for $\rere(\re(\Pi_{\Delta}(a,x)))$: each node simply replaces the node configuration it outputs in the given solution by a relaxation thereof that is contained in the node configuration of $\Pi_{\operatorname{rel}}$ (such that each set in the configuration is replaced by a superset thereof).
	This is possible since, as we proved, any node configuration of $\rere(\re(\Pi_{\Delta}(a,x)))$ can be relaxed to some node configuration of $\Pi_{\operatorname{rel}}$.
	Note that the fact that the $0$-round algorithm replaces sets by supersets, together with the observations about the aforementioned method, ensures that the produced edge configurations are indeed contained in the edge constraint of $\Pi_{\operatorname{rel}}$.
	
	The desired statement that $\Pi^+_{\Delta}(a,x)$ can be solved in $0$ rounds given an arbitrary solution for $\rere(\re(\Pi_{\Delta}(a,x)))$ (and hence the lemma statement) now follows from the observation that the two problems $\Pi^+_{\Delta}(a,x)$ and $\Pi_{\operatorname{rel}}$ are the same up to renaming as shown by the following mapping.
	\begin{align*}
		\mybox{\zx\zo\za\zp} &\mapsto \M \\
		\mybox{\xx\zx\xo\zo\xa\za\xp\zp} &\mapsto \X \\
		\mybox{\xp\zp} &\mapsto \P \\
		\mybox{\xo\zo\xa\za\xp\zp} &\mapsto \O \\
		\mybox{\xa\za\xp\zp} &\mapsto \A \\
		\mybox{\zo\za\xp\zp} &\mapsto \C \\
	\end{align*}

\end{proof}

We now prove that, by exploiting a given $\Delta$-edge coloring, we can relate  $\Pi^+_{\Delta}(a,x)$ with some problem of the family.

\begin{lemma}\label{lem:use-edge-col}
	If we are given a $\Delta$-edge coloring in input, then $\Pi^+_{\Delta}(a,x)$ is at least as hard as $\Pi_{\Delta}(\lfloor(a-2x-1)/2\rfloor,x+1)$, for all $x,a$ satisfying $2x+1 \le a \le \Delta$.
\end{lemma}
\begin{proof}
	We show how to use any given solution for $\Pi^+ = \Pi^+_{\Delta}(a,x)$ to solve, in $0$ rounds, the problem $\Pi = \Pi_{\Delta}(\lfloor(a-2x-1)/2\rfloor,x+1)$.
	Note that $\Pi^+$ and $\Pi$ are defined similarly. In fact, they only differ in two places: 
	\begin{itemize}
		\item In $\Pi$ the label $\C$ never appears, hence in order to solve $\Pi$ given a solution for $\Pi^+$, we need to change the output of all nodes outputting $\C^{\Delta-x} \X^x$.
		\item In $\Pi^+$ the nodes outputting the configuration containing the label $\A$ are required to own $a-x-1$ edges, while in $\Pi$ this number is only $\lfloor(a-2x-1)/2\rfloor$.
	\end{itemize}
	We show how to exploit a given $\Delta$-edge coloring to change the output of nodes outputting the configuration containing the label $\C$ to an output that uses the configuration containing $\A$. Note that if all nodes outputting $\C^{\Delta-x} \X^x$ choose $\lfloor(a-2x-1)/2\rfloor$ arbitrary incident edges that are labeled $\C$ and label them with $\A$ and all the others with $\X$, this would almost work: $\C$ is edge-compatible with $[\M\A\O\X]$, and if we ignore $\A$, all other labels are also compatible with $\A$. The only issue is given by edges labeled $\A \s \C$: if such a $\C$ is converted to an $\A$, we obtain an edge labeled $\A \s \A$, that is not allowed. Here we exploit the given $\Delta$-edge coloring. 
	\begin{itemize}
		\item All nodes labeled $\A^{a-x-1}\s\X^{\Delta-a+1+x}$ start by replacing label $\A$ with label $\X$ on the incident edges labeled with colors in $\{1,\ldots,\lfloor(a-1)/2 \rfloor\}$. Note that the remaining number of $\A$ is at least $a-x-1 - \lfloor (a-1)/2 \rfloor = \lceil (a-2x-1)/2\rceil$. Then, they replace arbitrary other $\A$ with $\X$ in order to make the number of $\A$ exactly $\lfloor (a-2x-1)/2 \rfloor$. 
		
		\item All nodes labeled $\C^{\Delta-x}\s\X^{x}$ consider their incident edges labeled with colors in $\{1,\ldots,\lfloor(a-1)/2 \rfloor\}$ that are currently labeled $\C$, write $\A$ on them and $\X$ on all the others. Note that the number of $\A$ is at least $\lfloor(a-1)/2 \rfloor - x = \lfloor(a-2x-1)/2 \rfloor$. Finally, they replace some of the written $\A$ with $\X$ in order to make the number of $\A$ exactly $\lfloor (a-2x-1)/2 \rfloor$.
	\end{itemize}
	Note that there is no need of coordination between nodes, and hence this procedure requires $0$ rounds. We now prove that the obtained labeling is a solution for $\Pi$. 
	\begin{itemize}
		\item Nodes that were labeled with the configuration containing $\A$, replaced some $\A$ with $\X$ to match the required number of $\A$ for $\Pi$, and since $\X$ is edge-compatible with anything that is compatible with $\A$, this part does not violate any constraint.
		\item It cannot happen that an edge is labeled $\A \s \A$, since nodes that were labeled with the configuration containing $\C$ only write the label $\A$ on edges that are colored with $\{1,\ldots,\lfloor(a-1)/2 \rfloor\}$, and on the same edges the nodes originally labeled with the configuration containing $\A$ replaced the label $\A$ with $\X$. Also, since $\C\s \C$ is not allowed, it cannot happen that two neighboring nodes that are labeled with the configuration containing $\C$ both write $\A$ on their common edge (note that they may be neighbors through edges labeled $\X$).
		\item $\C$ is edge compatible with $[\M\A\O\X]$. In the replacing process, some $\C$ are replaced with $\A$ and some others with $\X$, and since we proved that $\A\s\A$ is never obtained, we may only obtain edges labeled $[\A\X] \s [\M\O\X]$, which is allowed by $\Pi$.
		\item All other configurations present in the constraints of $\Pi^+$ are also allowed by the constraints of $\Pi$.
	\end{itemize}
\end{proof}

By combining Lemma \ref{lem:speedup} and Lemma \ref{lem:use-edge-col} we get the following corollary.
\begin{corollary}\label{cor:onestep}
	If the problem $\Pi_{\Delta}(a,x)$ has time complexity $T$, then  $\Pi_{\Delta}(\lfloor(a-2x-1)/2\rfloor,x+1)$ has complexity at most $\max\{T-1,0\}$, for all $a, x$ satisfying $2x+1 \le a$, $x+2 \leq a \leq \Delta$, given a $\Delta$-edge coloring in input.
\end{corollary}

In the above we have shown a relation between $\Pi_{\Delta}(a,x)$ and $\Pi_{\Delta}(\lfloor(a-2x-1)/2\rfloor,x+1)$. Note that there is also a more straightforward relation between the problems of the family, that is, by increasing $x$ or by decreasing $a$, the problem does not get harder.
\begin{lemma}\label{lem:makeeasier}
	$\Pi_{\Delta}(a,x)$ can be solved in $0$ rounds given a solution of $\Pi_{\Delta}(a',x')$, for all $a \le a'$ and $x \ge x'$. 
\end{lemma}
\begin{proof}
	In order to convert a solution for $\Pi_{\Delta}(a',x')$ into a solution of $\Pi_{\Delta}(a,x)$ nodes have to relabel some edges labeled $\M$ and $\A$ with $\X$. Since $\X$ is compatible with everything, the edge constraint is not violated after performing this process.
\end{proof}

\subsection{Zero Rounds Solvability}
We now prove that some problems of the family are not $0$-rounds solvable, even if we are given a $\Delta$-edge coloring in input. 
\begin{lemma}\label{lem:zeropn}
	The problem $\Pi_{\Delta}(a,x)$ cannot be solved deterministically in $0$ rounds in the port numbering model, for all $x \le \Delta-1$ and $a \ge 1$, even if a $\Delta$-edge coloring is given in input.
\end{lemma}
\begin{proof}
	Consider a family of graphs where ports are numbered such that edges of color $i$ have port number $i$ assigned for both the endpoints, for all $i$. Note that the $0$-rounds view of an algorithm for the port numbering model running in this family of graphs is the same for all nodes (while nodes can see the port numbers connecting them to their incident edges, they do not even see the port numbering of the edges, that is, their orientation). This means that a deterministic $0$-rounds algorithm must output the same configuration for all nodes. Such an algorithm is allowed to decide the mapping between the $\Delta$ labels of the configuration and the $\Delta$ ports, but due to the nature of our graph family, each edge obtains the same label on both endpoints (since each edge has the same ports assigned on both sides). We now consider all possible configurations allowed by the node constraint of $\Pi_{\Delta}(a,x)$ and show that if an algorithm tries to use any of them, then it must fail. For this purpose, we show that for all allowed configurations there is at least one label that is not edge-compatible with itself. For the configuration $\M^{\Delta-x} \X^x$ we pick label $\M$, for the configuration  $\A^{a} \s \X^{\Delta - a}$ we pick label $\A$, and for the configuration $\P \s \O^{\Delta - 1}$ we pick label $\P$. Note that $\M \s \M$, $\A \s \A$, and $\P \s \P$, are all configurations not allowed by the edge constraint.
\end{proof}

\subsection{Lower Bound for the Port Numbering Model}
We now put things together and obtain an $\Omega(\log \Delta)$ deterministic lower bound for $k$-outdegree dominating sets in the port numbering model, for all $k$ that are not too large compared to $\Delta$.

\begin{lemma}\label{lem:sequence}
	Let $t = \epsilon \log \Delta$ and $x \le \Delta^{\epsilon}$, for some constant $\epsilon > 0$. For any large enough $\Delta$, there exists a sequence of problems $\Pi_0 \rightarrow \Pi_1 \rightarrow \ldots \rightarrow \Pi_t$ such that, given a $\Delta$-edge coloring,
	\begin{enumerate}
		\item $\Pi_0 = \Pi_{\Delta}(\Delta,x)$, 
		\item $\Pi_{i+1}$ can be solved in $0$ rounds given a solution to $\rere(\re(\Pi_i))$, for all $0 \le i \le t-1$,
		\item $\Pi_{t}$ cannot be solved in $0$ rounds in the deterministic port numbering model,
		\item all problems in the sequence use at most $5$ labels.
	\end{enumerate}
\end{lemma}
\begin{proof}
	For all $i \ge 0$, we define $\Pi_{i} = \Pi_{\Delta}(\lfloor \Delta /2^{3i} \rfloor,x+i)$. Note that, point 4 is trivially true, and also, for $i=0$, this definition matches the definition given in point 1. We now prove point 2, that is that $\Pi_{i}$ is at least one round harder than $\Pi_{i+1}$. 
	
	Let $\Pi_i=\Pi_\Delta(\bar{a},\bar{x})$ be an arbitrary problem of the sequence, where $\bar{a}= \lfloor \Delta /2^{3i} \rfloor$ and $\bar{x}=x+i$. Notice that $\bar{x} < \bar{a}/8$, since for $x \le \Delta^\epsilon$, $i \le \epsilon \log \Delta$, and $\epsilon$ small enough, $\bar{x}=x + i < \lfloor \Delta /2^{3i} \rfloor/8=\bar{a}/8$. By applying Corollary \ref{cor:onestep} we get that $\Pi_i$ is at least one round harder than $\Pi'=\Pi_{\Delta}(\lfloor(\bar{a}-2\bar{x}-1)/2\rfloor,\bar{x}+1)$. Assume $\bar{a}\ge 4$, which holds for all the problems in the sequence if $\Delta$ is large enough. By applying Lemma \ref{lem:makeeasier}, since $\bar{x} < \bar{a}/8$, we get that $\Pi'$ is at least as hard as $\Pi_{\Delta}(a',\bar{x}+1)$, where $a' = \lfloor \bar{a}/4 \rfloor = \lfloor \frac{\lfloor \Delta /2^{3i} \rfloor}{4} \rfloor \ge \lfloor \frac{ \Delta /2^{3i} -1}{4} \rfloor $, that is at least $\lfloor \frac{ \Delta /2^{3i}}{8} \rfloor  = \lfloor  \Delta /2^{3(i+1)} \rfloor $. Hence, $\Pi'$ is at least as hard as $\Pi_\Delta(\lfloor  \Delta /2^{3(i+1)} \rfloor, \bar{x} + 1)=\Pi_{i+1}$. This implies that $\Pi_i$ is at least one round harder than $\Pi_{i+1}$, proving point 2.
	
	For all $x \le \Delta^\epsilon$, $\Pi_t$ is not easier than $\Pi_{\Delta}(\lfloor \Delta /\Delta^{3\epsilon} \rfloor,\Delta^{\epsilon}+\epsilon \log \Delta)$, that for small enough $\epsilon$ and large enough $\Delta$, by Lemma \ref{lem:zeropn}, is not $0$ rounds solvable, proving point 3.
\end{proof}

\section{Lower bound for the LOCAL model}\label{sec:local}
By using standard techniques, we can lift the $\Omega(\log \Delta)$ lower bound obtained for the deterministic port numbering model, to a lower bound for the \LOCAL model. We note that the techniques used in previous works for achieving this goal do not actually depend on the specific problem $\Pi$ for which we want to lift the bound, but just on some properties of $\Pi$. Hence, from previous works, the following is known.
\begin{theorem}[\cite{balliurules,trulytight,Balliu2019}]\label{thm:known}
	Let $\Pi_0 \rightarrow \Pi_1 \rightarrow \ldots \rightarrow \Pi_t$ be a sequence of problems such that $\Pi_{i+1}$ can be solved in $0$ rounds given a solution for $\rere(\re(\Pi_i))$, the number of labels of each problem $\Pi_i$ is upper bounded by $O(\Delta^2)$, and for all $t' < t$, $\Pi_{t'}$ is not $0$-round solvable in the randomized port numbering model with failure probability smaller than $1/\Delta^8$ even if ports are assigned such that for any edge $\{u,v\}$ it holds that if $u$ is connected to $v$ through port $i$, then also $v$ is connected to $u$ through port $i$. Then, $\Pi_0$ requires $\Omega(\min\{t, \log_\Delta n\})$ in the deterministic \LOCAL model and $\Omega(\min\{t, \log_\Delta \log n\})$ in the randomized \LOCAL model.
\end{theorem}

We now prove a lower bound on the failure probability of any algorithm in the randomized port numbering model that tries to solve $\Pi_\Delta(a,x)$ in $0$ rounds given a $\Delta$-edge coloring in input.
\begin{lemma}\label{lem:zerolocal}
	The problem $\Pi_\Delta(a,x)$ is not $0$-round solvable with failure probability less than $1/\Delta^8$ in the randomized port numbering model, for all $x \le \Delta-1$ and $a \ge 1$, even if a $\Delta$-edge coloring is given in input.
\end{lemma}
\begin{proof}
	Consider a family of graphs where ports are numbered such that edges of color $i$ have the port number $i$ assigned for both the endpoints, for all $i$. Note that the $0$-rounds view of a randomized algorithm for the port numbering model, running in this family of graphs, is the same for all nodes, except for their random bits. This means that a randomized $0$-rounds algorithm must output the same configuration for all nodes with the same probability.
	Since the node constraint of $\Pi_\Delta(a,x)$ contains only $3$ allowed configurations, then there is at least one configuration that is used by all nodes with probability at least $1/3$. As stated in the proof of Lemma \ref{lem:zeropn}, all configurations contain at least one label $\ell$ that is not edge-compatible with itself. This means that there must be a port $i$ where all nodes with probability at least $1/(3\Delta)$ use label $\ell$. Since $\ell$ is not compatible with itself, a $0$-round algorithm fails with probability at least $1/(3\Delta)^2 \ge 1 / \Delta^8$.
	 
\end{proof}

By combining Lemma \ref{lem:tofamily} (the relation between $k$-outdegree dominating sets and the problems of the family), Lemma \ref{lem:sequence} (the existence of a long sequence of problems), Theorem \ref{thm:known} (the existence of a long sequence of problems implies a lower bound for the \LOCAL model, given a randomized lower bound for the port numbering model), and Lemma \ref{lem:zerolocal} (a randomized lower bound for the port numbering model), we obtain our main result, Theorem \ref{thm:main}, which is stated in Section \ref{sec:results}.
By setting $\Delta \approx 2^{\sqrt{\log n}}$ in the deterministic case and $\Delta \approx 2^{\sqrt{\log \log n}}$ in the randomized case, we directly obtain Corollary \ref{cor:main}, which is also stated in Section \ref{sec:results}.
\

\section{Open Problems}
We proved that solving the maximal independent set problem on trees requires $\Omega(\min\{\log \Delta,f(n)\})$, where $f(n) = \sqrt{\log n}$ for deterministic algorithms, and $f(n) = \sqrt{\log \log n}$ for randomized ones, and that the same holds even for more relaxed variants of the problem. An interesting open question would be to understand the right dependency on $\Delta$, and we conjecture it to be $\Omega(\Delta)$.

One of the reasons why we believe it is important to understand the right dependency on $\Delta$ in trees is related to ruling sets. For $(2,k)$-ruling sets, even if a $\Delta+1$ coloring is already provided in input, the best known upper bound is $O(\Delta^{1/k})$, while the best known lower bound lies in the polylogarithmic in $\Delta$ region. Unfortunately, the $\Omega(\Delta)$ lower bound that is known for MIS is proved on line graphs, where ruling sets are easy to solve \cite{KuhnMW18}, and hence the MIS proof cannot be extended to ruling sets. For this reason, we believe that proving an $\Omega(\Delta)$ lower bound for MIS on trees would help in proving improved lower bounds for ruling sets. 

In this work we managed to improve the lower bound for MIS on trees by exploiting a given $\Delta$-edge coloring. The same technique seems to \emph{fail} for ruling sets. An open question is understanding what other kind of inputs can be effective for performing simplifications in the round elimination framework.

\bibliographystyle{ACM-Reference-Format}
\bibliography{mis-lower-bound}

\end{document}